\documentclass[aps,prl,groupedaddress,superscriptaddress,twocolumn,notitlepage,bibnotes]{revtex4-1}

\usepackage{cancel}
\usepackage[latin1]{inputenc}
\usepackage{amsthm}
\usepackage{amssymb}
\usepackage{amsmath}
\usepackage{bbold}
\usepackage{bbm}
\usepackage{nonfloat}
\usepackage[pdftex]{hyperref}
\usepackage{braket}
\usepackage{dsfont}
\usepackage{mathdots}
\usepackage{mathtools}
\usepackage{enumerate}
\usepackage[shortlabels]{enumitem}
\usepackage{csquotes}
\usepackage{stmaryrd}
\usepackage[cal=boondox]{mathalfa}
\usepackage{graphicx}
\usepackage{stackengine}
\usepackage{scalerel}
\usepackage{xr}
\usepackage{array}
\usepackage{makecell}
\newcolumntype{x}[1]{>{\centering\arraybackslash}p{#1}}
\usepackage{tikz}
\usepackage{pgfplots}
\usetikzlibrary{shapes.geometric, shapes.misc, positioning, arrows, arrows.meta, decorations.pathreplacing, decorations.pathmorphing, patterns, angles, quotes, calc}
\usepackage{booktabs}
\usepackage{xfrac}
\usepackage{siunitx}
\usepackage{centernot}
\usepackage{comment}
\usepackage{chngcntr}

\newtheorem{thm}{Theorem}
\newtheorem*{thm*}{Theorem}

\newtheorem*{prop*}{Proposition}
\newtheorem{lemma}[thm]{Lemma}
\newtheorem*{lemma*}{Lemma}

\newtheorem*{cor*}{Corollary}
\newtheorem{cj}[thm]{Conjecture}
\newtheorem*{cj*}{Conjecture}
\newtheorem{Def}[thm]{Definition}
\newtheorem*{Def*}{Definition}

\makeatletter
\def\thmhead@plain#1#2#3{%
  \thmname{#1}\thmnumber{\@ifnotempty{#1}{ }\@upn{#2}}%
  \thmnote{ {\the\thm@notefont#3}}}
\let\thmhead\thmhead@plain
\makeatother

\theoremstyle{definition}

\newtheorem{manualthminner}{Theorem}
\newenvironment{manualthm}[1]{%
  \renewcommand\themanualthminner{#1}%
  \manualthminner \it
}{\endmanualthminner}

\newcommand{\bb}{\begin{equation}\begin{aligned}\hspace{0pt}}
\newcommand{\bbb}{\begin{equation*}\begin{aligned}}
\newcommand{\ee}{\end{aligned}\end{equation}}
\newcommand{\eee}{\end{aligned}\end{equation*}}
\newcommand*{\coloneqq}{\mathrel{\vcenter{\baselineskip0.5ex \lineskiplimit0pt \hbox{\scriptsize.}\hbox{\scriptsize.}}} =}

\newcommand{\ketbra}[1]{\ket{#1}\!\!\bra{#1}}
\newcommand{\ketbraa}[2]{\ket{#1}\!\!\bra{#2}}

\newcommand{\ketbrasub}[1]{\ket{#1}\!\bra{#1}}

\newcommand{\tcb}[1]{{\color{blue} #1}}

\newcommand{\ludo}[1]{{\color{green!60!black} #1}}

\newcommand{\R}{\mathds{R}}
\newcommand{\N}{\mathds{N}}

\DeclareMathOperator{\Tr}{Tr}

\DeclareMathAlphabet{\pazocal}{OMS}{zplm}{m}{n}

\DeclareMathOperator{\Id}{Id}

\newcommand{\HH}{\pazocal{H}}

\newcommand{\lsmatrix}{\left(\begin{smallmatrix}}
\newcommand{\rsmatrix}{\end{smallmatrix}\right)}

\stackMath

\stackMath

\makeatletter
\newcommand*\rel@kern[1]{\kern#1\dimexpr\macc@kerna}
\newcommand*\widebar[1]{%
  \begingroup
  \def\mathaccent##1##2{%
    \rel@kern{0.8}%
    \overline{\rel@kern{-0.8}\macc@nucleus\rel@kern{0.2}}%
    \rel@kern{-0.2}%
  }%
  \macc@depth\@ne
  \let\math@bgroup\@empty \let\math@egroup\macc@set@skewchar
  \mathsurround\z@ \frozen@everymath{\mathgroup\macc@group\relax}%
  \macc@set@skewchar\relax
  \let\mathaccentV\macc@nested@a
  \macc@nested@a\relax111{#1}%
  \endgroup
}

\counterwithin*{equation}{part}
\counterwithin*{thm}{part}
\counterwithin*{figure}{part}

\tikzset{meter/.append style={draw, inner sep=10, rectangle, font=\vphantom{A}, minimum width=30, line width=.8, path picture={\draw[black] ([shift={(.1,.3)}]path picture bounding box.south west) to[bend left=50] ([shift={(-.1,.3)}]path picture bounding box.south east);\draw[black,-latex] ([shift={(0,.1)}]path picture bounding box.south) -- ([shift={(.3,-.1)}]path picture bounding box.north);}}}
\tikzset{roundnode/.append style={circle, draw=black, fill=gray!20, thick, minimum size=10mm}}
\tikzset{squarenode/.style={rectangle, draw=black, fill=none, thick, minimum size=10mm}}

\definecolor{Blues5seq1}{RGB}{239,243,255}
\definecolor{Blues5seq2}{RGB}{189,215,231}
\definecolor{Blues5seq3}{RGB}{107,174,214}
\definecolor{Blues5seq4}{RGB}{49,130,189}
\definecolor{Blues5seq5}{RGB}{8,81,156}

\definecolor{Greens5seq1}{RGB}{237,248,233}
\definecolor{Greens5seq2}{RGB}{186,228,179}
\definecolor{Greens5seq3}{RGB}{116,196,118}
\definecolor{Greens5seq4}{RGB}{49,163,84}
\definecolor{Greens5seq5}{RGB}{0,109,44}

\definecolor{Reds5seq1}{RGB}{254,229,217}
\definecolor{Reds5seq2}{RGB}{252,174,145}
\definecolor{Reds5seq3}{RGB}{251,106,74}
\definecolor{Reds5seq4}{RGB}{222,45,38}
\definecolor{Reds5seq5}{RGB}{165,15,21}
\usepackage{pgfplots}

\pgfplotsset{width=10cm,compat=1.9}
\usetikzlibrary{decorations.markings}

\begin{document}


\title{Restoring quantum communication efficiency over high loss optical fibres}

\author{Francesco Anna Mele}
\email{francesco.mele@sns.it}
\affiliation{NEST, Scuola Normale Superiore and Istituto Nanoscienze, Consiglio Nazionale delle Ricerche, Piazza dei Cavalieri 7, IT-56126 Pisa, Italy}

\author{Ludovico Lami}
\email{ludovico.lami@gmail.com}
\affiliation{Institut f\"{u}r Theoretische Physik und IQST, Universit\"{a}t Ulm, Albert-Einstein-Allee 11, D-89069 Ulm, Germany}

\author{Vittorio Giovannetti}
\email{vittorio.giovannetti@sns.it}
\affiliation{NEST, Scuola Normale Superiore and Istituto Nanoscienze, Consiglio Nazionale delle Ricerche, Piazza dei Cavalieri 7, IT-56126 Pisa, Italy}

\begin{abstract}
In the absence of quantum repeaters, quantum communication proved to be nearly impossible across optical fibres longer than $\gtrsim \SI{20}{\km}$ due to the drop of transmissivity below the critical threshold of $1/2$. However, if the signals fed into the fibre are separated by a sufficiently short time interval, memory effects must be taken into account. In this paper we show that by properly accounting for these effects it is possible to devise schemes that enable unassisted quantum communication across arbitrarily long optical fibres at a fixed positive qubit transmission rate. We also demonstrate how to achieve entanglement-assisted communication over arbitrarily long distances at a rate of the same order of the maximum achievable in the unassisted noiseless case.
\end{abstract}

\maketitle

\textbf{\em Introduction}.--- 
Reliably transmitting qubits across long optical fibres is crucial for establishing a global quantum internet~\cite{quantum_internet_Wehner, quantum_internet_kimble}
, an architecture that would allow unconditionally secure communication, entanglement distribution over long distances, quantum sensing improvements (e.g.~in telescope observations~\cite{Telescope_qinternet} and clock synchronisation~\cite{Clock_qinternet}), distributed quantum computing~\cite{Distributed_QC}
and private remote access to quantum computers~\cite{secure_access_qinternet}.
The main technological hurdle in establishing a global quantum internet is the fact that the transmissivity of an optical fibre decreases rapidly --- typically, exponentially --- with its length. 
This is generally a serious problem, because the (unassisted) quantum capacity of an optical fibre --- i.e.\ its ability to reliably transmit quantum messages --- drops to zero if the overall transmissivity of the communication line falls below the critical value of $1/2$. 
A similar effect applies also to the two-way assisted quantum capacity of these models --- i.e.\ their ability to reliably transmit quantum information with the help of free classical communication --- which is known to vanish for sufficiently low transmissivities~\cite{Caruso2006}. Typical optical fibres employed nowadays attenuate the signal by \SI[per-mode = symbol]{0.2}{\decibel\per\km}, the absolute record being \SI[per-mode = symbol]{0.14}{\decibel\per\km}~\cite{Tamura2018, Li2020}; this means that in absence of quantum repeaters~\cite{repeaters, Munro2015, quantum_repeaters_linearopt} the quantum capacity vanishes for fibres longer than \SI{15}{\km} or at most \SI{21.5}{\km}. 


Here, we present a conceptually simple scheme that overcomes this problem and in principle enables quantum communication at a constant rate over arbitrarily long distances, i.e.\ for arbitrarily low non-zero values of the transmissivity. We further prove that this same scheme can be combined with entanglement assistance to neutralise the effect of noise in classical communication altogether. We achieve these results by studying information transmission preceded by a trigger pulse: if the time interval between trigger pulse and signal is sufficiently short, then the memoryless assumption commonly invoked in the quantum capacity analysis of these models breaks down, and one can effectively alter the environment before the actual transmission begins~\cite{memory-review, Dynamical-Model}. Since memory effects have been experimentally observed in optical fibres~\cite{Banaszek-Memory,Ball-Memory}, the above scheme may offer a concrete route to enable quantum communication over long distances.



At the quantum level an optical fiber is typically described by a  memoryless thermal attenuator $\Phi_{\lambda,\tau_\nu}$, i.e.\ a bosonic quantum channel~\cite{HOLEVO-CHANNELS-2} that mixes the input signal with a thermal environment $\tau_\nu$ through a beam splitter (BS) of transmissivity $\lambda\in[0,1]$. Our findings build on those of~\cite{die-hard}, where the phenomenon of `die-hard quantum communication' (D-HQCOM) was uncovered. Such an effect consists in the observation that if we replace the thermal state $\tau_\nu$ of the environment with a suitable 
state $\sigma=\sigma(\lambda)$ (possibly dependent upon the transmissivity of the model), the quantum capacity of the modified attenuator channel $\Phi_{\lambda,\,\sigma(\lambda)}$ stays above a fixed positive constant $c>0$. The reason why this result cannot be applied directly to improve the quantum communication capabilities of long optical fibres is that neither the sender (Alice) nor the receiver (Bob) can realistically have access to the initial state of the environment. Prior to the present work, D-HQCOM thus seemed mostly a mathematical oddity. Overturning this view, our main conceptual contribution is to recognise that it can instead be turned into a potentially viable technology by exploiting memory 
effects, paving new avenues to quantum communication beyond current limitations.


\smallskip \textbf{\em Notation}.--- 
The set of quantum states 
on a Hilbert space $\HH$ 
is denoted by $\mathfrak{S}(\HH)$. Given $\rho_1,\rho_2\in\mathfrak{S}(\HH)$, the trace norm of their difference is denoted as $\|\rho_1-\rho_2\|_1$. 
The information-carrying signal we consider is a single-mode of electromagnetic radiation with definite frequency and polarisation. This system, associated with the Hilbert space $\HH_S\coloneqq L^2(\mathbb{R})$, is described as a quantum harmonic oscillator.
A BS of transmissivity $\lambda\in[0,1]$ acting on two single-mode systems  $S_1$ and $S_2$ is defined as the unitary transformation $U_{\lambda}^{S_1 S_2}\coloneqq\exp\left[\arccos\sqrt{\lambda}\left(a_1^\dagger a_2-a_1a_2^\dagger\right)\right]$, where $a_1$ and $a_2$ are the annihilation operators on $S_1$ and $S_2$, respectively. 
Let $\HH_E\coloneqq L^2(\mathbb{R})$ denote a single-mode system, dubbed `environment', and let $b$ denote its annihilation operator. Fixed $\lambda\in[0,1]$ and $\sigma\in\mathfrak{S}(\HH_E)$, a general attenuator $\Phi_{\lambda,\sigma}:\mathfrak{S}(\HH_S)\longmapsto\mathfrak{S}(\HH_S)$ is a quantum channel defined by $\Phi_{\lambda,\sigma}(\rho)\coloneqq\Tr_E\left[U_\lambda^{SE}\text{} \rho\otimes\sigma\text{} \left(U_\lambda^{SE}\right)^\dagger\right]$.
$\Phi_{\lambda,\sigma}$ is completely noisy for $\lambda=0$, in the sense that $\Phi_{0,\sigma}(\rho)=\sigma$ for all $\rho$; on the contrary, it is noiseless, i.e.\ it coincides with the identity channel $\Id$, for $\lambda=1$. 
As mentioned in the introduction, a thermal attenuator $\Phi_{\lambda,\tau_\nu}$ is a special case of general attenuator obtained by identifying $\sigma$ with a thermal state $\tau_\nu\coloneqq \frac{1}{\nu+1}\sum_{n=0}^\infty \left( \frac{\nu}{\nu+1}\right)^n \ketbra{n}$, with $\ket{n}$ being the $n$th Fock state of the model.


The energy-constrained (EC) classical capacity $C(\Phi,N)$ (resp.\ quantum capacity $Q(\Phi,N)$) of a quantum channel $\Phi$ is the maximum rate of bits (resp.\ qubits) that can be reliably transmitted through $\Phi$ assuming that Alice has access to a limited amount ($N$) of input energy to build her coding signals. In addition, if an unlimited amount of pre-shared entanglement can be exploited by Alice and Bob, the corresponding maximum achievable bit (resp.\ qubit) transmission rate is called the EC entanglement-assisted (EA) classical (resp.\ quantum) capacity $C_{\text{ea}}(\Phi,N)$ (resp.\ $Q_{\text{ea}}(\Phi,N)$)~\cite{Bennett-EA,Bennett2002,Mark-energy-constrained, Holevo2004, Holevo-energy-constrained, Holevo2013} --- the two quantities being related by the identity $C_{\text{ea}}(\Phi,N)=2Q_{\text{ea}}(\Phi,N)$ thanks to quantum teleportation~\cite{teleportation} and super-dense coding~\cite{dense-coding}. 
The EC EA classical capacity can be expressed as~\cite{entanglement-assisted,Bennett2002, Holevo2013}:
\begin{equation}\label{Cea}
C_{\text{ea}}(\Phi,N)=\max_{\rho\in\mathfrak{S}(\HH_S):\\\text{}\Tr\left[\rho a^\dagger a \right]\le N}\left[S(\rho)+I_{\mathrm{coh}}\left(\Phi,\rho\right)\right],
\end{equation}	
where $S(\rho)\coloneqq -\Tr[\rho\log_2\rho]$, 
$I_{\mathrm{coh}}\left(\Phi,\rho\right)\coloneqq S\left(\Phi(\rho)\right)-S\left(\Phi\otimes \Id_P(\ketbra{\psi})\right)$, 
with $\ket{\psi}\in\HH_S\otimes\HH_{P}$ being a purification of $\rho$~\cite{NC,MARK,HOLEVO-CHANNELS-2}, $\HH_P$ being the purifier Hilbert space, and $\Id_P$ being the identity superoperator on $\HH_P$.
In addition, the EC quantum capacity 
can be written as~\cite{Lloyd-S-D, L-Shor-D, L-S-Devetak, HOLEVO-CHANNELS-2,Mark-energy-constrained}:
\bb\nonumber
Q(\Phi,N)=&\ \lim\limits_{n\to\infty}\frac{1}{n}Q_1(\Phi^{\otimes n},nN )\ge Q_1(\Phi,N)\,,\\
Q_1\!\left(\Phi^{\otimes n}\!,N\right)\coloneqq&\ \max_{\rho\in\mathfrak{S}(\HH_S^{\otimes n}):\text{}\Tr\left[\rho \sum_{i=1}^na_i^\dagger a_i \right]\le N}I_{\mathrm{coh}}\left(\Phi^{\otimes n}\!,\rho\right)\,.
\ee
The quantities $C_{\text{ea}}\left(\Phi_{\lambda,\tau_{\nu}},N\right)$~\cite{holwer, LossyECEAC1, LossyECEAC2}, $C\left(\Phi_{\lambda,\tau_{\nu}},N\right)$~\cite{Giova_classical_cap}, and $Q\left(\Phi_{\lambda,\ketbrasub{0}},N\right)$~\cite{holwer, Caruso2006, Wolf2007, Mark2012, Mark-energy-constrained, Noh2019} have been determined exactly. For $\nu>0$, sharp bounds on $Q\left(\Phi_{\lambda,\tau_{\nu}},N\right)$ are known if $\lambda> 1/2+\frac{\nu}{2(\nu+1)}$~\cite{PLOB, Rosati2018, Sharma2018, Noh2019,holwer, Noh2020,fanizza2021estimating}, while it is known that $Q\left(\Phi_{\lambda,\tau_{\nu}},N\right)=0$ otherwise~\cite{Rosati2018}.

\smallskip
\textbf{\em Environment control and EA imply noise neutralisation}.--- In~\cite{die-hard} it was shown that, even for arbitrarily low values of the transmissivity $\lambda>0$, suitable choices of the environmental state make the quantum capacity of the corresponding general attenuator bounded away from zero. Specifically:
\begin{thm}[\cite{die-hard}]\label{diehard_th}
For all $\lambda\in(0,1]$ there exists $\sigma(\lambda)$ such that
\begin{equation}
    Q\left(\Phi_{\lambda,\sigma(\lambda)}\right)\ge Q\left(\Phi_{\lambda,\sigma(\lambda)},1/2\right) >\eta\,,
\end{equation}
where $\eta>0$ is a universal constant.  More specifically, for $\varepsilon\ge0$ sufficiently small and for all $\lambda\in(0,1/2-\varepsilon)$ it holds that $Q\left(\Phi_{\lambda,\ketbrasub{n_\lambda} }\right)\ge Q\left(\Phi_{\lambda,\ketbrasub{n_\lambda} },1/2\right) >c(\varepsilon)$, where $c(\varepsilon)\ge0$ is a constant with respect to $\lambda$ and $n_\lambda\in\N$ satisfies $1/\lambda-1\le n_\lambda\le 1/\lambda$. Moreover, it holds that $c(0)=0$, and $c(\bar{\varepsilon})\ge 5.133\times10^{-6}$ for an appropriate $\bar{\varepsilon}$ such that $0<\bar{\varepsilon}\ll 1/6$ (see the Supplemental Material of~\cite{die-hard}). 
\end{thm}

Here, we provide an extension of this result for the EC EA capacities. We start by observing that for all $N>0$, $\lambda\in[0,1]$, and $n\in\N$, by choosing $\tau_N$ as an ansatz for $\rho$ in~\eqref{Cea}, one obtains 
\begin{align}\label{bound_Cea}
    C_{\text{ea}}\left(\Phi_{\lambda,\ketbrasub{n}},N\right)\ge C(\Id,N)+I_{\mathrm{coh}}(\Phi_{\lambda,\ketbrasub{n}},\tau_N)\text{,}
\end{align}
where we used $S(\tau_N)=C(\Id,N)$ \footnote{It holds $C(\Id,N)=Q(\Id,N)=C_{\text{ea}}(\Id,N)/2=Q_{\text{ea}}(\Id,N)=S(\tau_N)$.}. The quantity $I_{\mathrm{coh}}(\Phi_{\lambda,\ketbrasub{n}},\tau_N)$ is calculated 
in Lemma~S9 in the SM~\footnote{See the Supplemental Material for complete proofs of some of the results discussed in the main text.} and expressed in a simple form by exploiting the `master equation trick' introduced in~\cite{Die-Hard-2-PRA}. By plotting it (see for example Fig.~\ref{Icoh figure}), one can notice that the lower endpoint of the $\lambda$-range for which $I_{\mathrm{coh}}(\Phi_{\lambda,\ketbrasub{n}},\tau_N)>0$ seems to converge to zero as $n$ grows. This leads to:
\begin{figure}[t]
\includegraphics[width=1.0\linewidth]{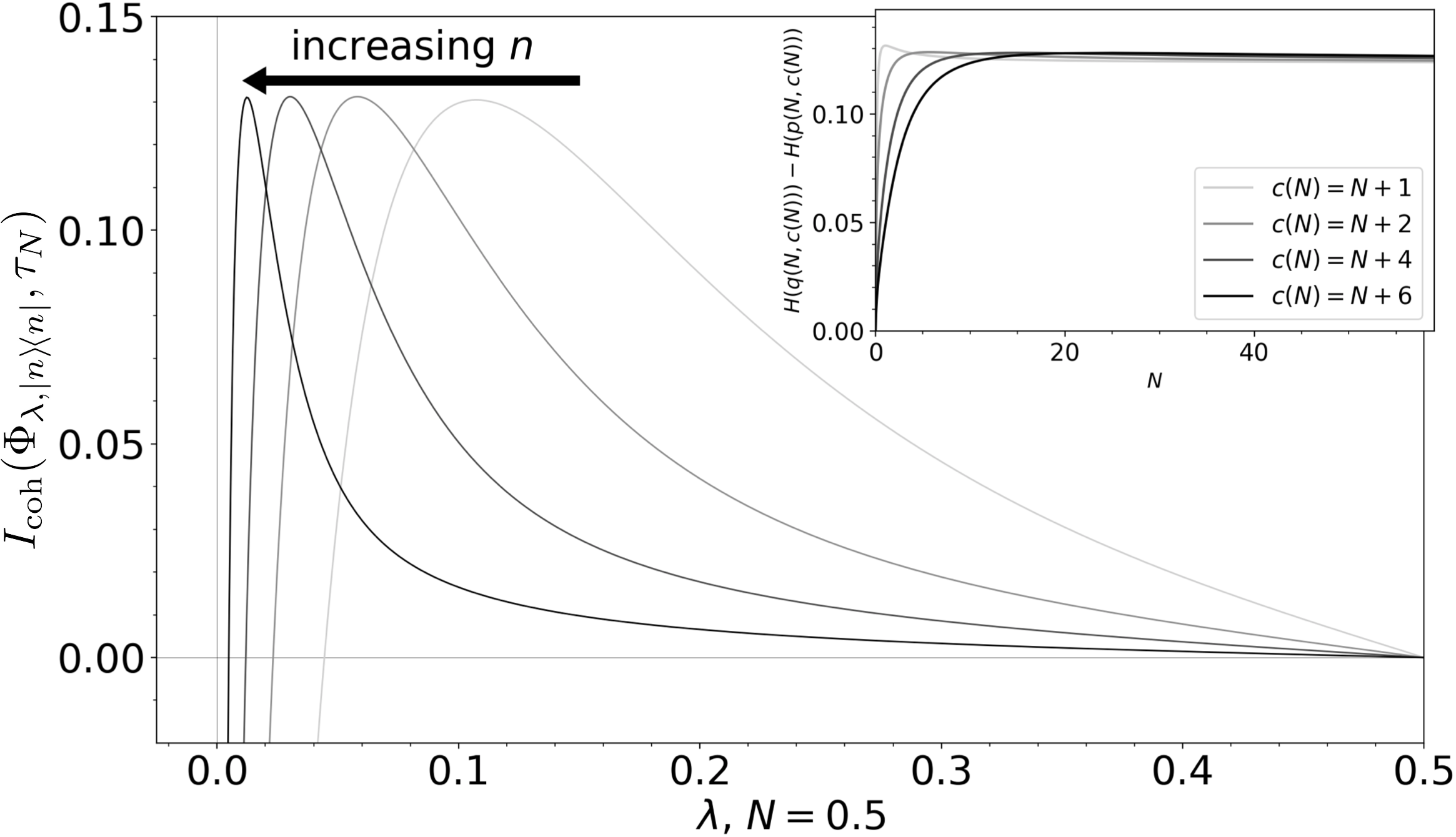}
\caption{The quantity $I_{\mathrm{coh}}(\Phi_{\lambda,\ketbrasub{n}},\tau_N)$ plotted with respect to $\lambda$ for $N=0.5$ and for several values of $n$ from $10$ to $100$. In the inset we plot the function $H\left(q(N,c(N))\right)-H\left(p(N,c(N))\right)$ (see Lemma~\ref{teo_lowtrasm}) with respect to $N$ for several choices of $c(N)$ of the form $c(N)=N+\alpha$.}
\label{Icoh figure}
\end{figure}
\begin{cj}\label{cong0}
For all $N>0$, $\lambda\in(0,1/2)$, if $n\in\N$ is sufficiently large, then $I_{\mathrm{coh}}\left(\Phi_{\lambda,\ketbrasub{n}},\tau_N\right)> 0$. 
\end{cj}

Our investigation (\cite[Section 3]{Die-Hard-2-PRA}, Lemma~\ref{teo_lowtrasm}, and Fig.~\ref{Icoh figure}) presents overwhelming numerical evidence that Conjecture~\ref{cong0} is true. Notice that if such statement is valid then~\eqref{bound_Cea} will imply that, 
irrespective of how small $\lambda$ is, by choosing $n$ sufficiently large one can make $C_{\text{ea}}\left(\Phi_{\lambda,\ketbrasub{n}},N\right)$ larger than or equal to the classical capacity of the noiseless channel $C(\Id,N)$ (and equivalently $Q_{\text{ea}}\left(\Phi_{\lambda,\ketbrasub{n}},N\right)>Q(\Id,N)/2$).
While we are not able to prove Conjecture~\ref{cong0} in full generality, in what follow we shall see that it holds at least in the most significant regime where $\lambda \to 0^+$. For this purpose we introduce the following~\cite{Note2}:

\begin{lemma}\label{teo_lowtrasm}
For all $N,c>0$ it holds that
\begin{align*}\label{entdifflow}
\liminf\limits_{n\to\infty}I_{\mathrm{coh}}\left(\Phi_{\frac{c}{n},\ketbrasub{n}},\tau_N\right) \ge H\left(q(N,c)\right)-H\left(p(N,c)\right)\,,
\end{align*}
where $\{q_{k}(N,c)\}_{k\in\mathbb{Z}}$ and $\{p_{k}(N,c)\}_{k\in\mathbb{Z}}$ are two probability distributions defined as
\begin{equation*}
\begin{aligned}
    q_k(N,c)&\coloneqq e^{-c(2N+1)}\left(\!\frac{N}{N\!+\!1}\!\right)^{k/2}I_{|k|}\!\left(2c\sqrt{N(N\!+\!1)}\right) ,\\
    p_k(N,c)&\coloneqq
        \begin{cases}
        \frac{e^{-c/(N+1)}N^k}{(N+1)^{k+1}}L_k\left(-\frac{c}{N(N+1)}\right) &\text{if $k\ge 0$ ,} \\
        0 &\text{otherwise.}
    \end{cases}
\end{aligned}
\end{equation*}
	with $I_k(\cdot)$ and $L_k(\cdot)$ being the $k$th Bessel function of the first kind and the $k$th Laguerre polynomial, respectively. $H(\cdot)$ denotes the Shannon entropy.
\end{lemma}
If we could find a value of $c$ that makes the quantity $H\left(q(N,c)\right)-H\left(p(N,c)\right)$ positive, then Conjecture~\ref{cong0} would be proved for $\lambda\to 0^+$. Let us choose for example $c(N)\coloneqq N+\alpha$, with $\alpha>0$ fixed.
The plot of the function $H\left(q(N,N+\alpha)\right)-H\left(p(N,N+\alpha)\right)$, shown in the inset of Fig.~\ref{Icoh figure}, demonstrates that such function is numerically verified to be positive for all $N$.
Consequently, thanks to~\eqref{bound_Cea}, we can conclude that
\begin{equation}\label{small_lambda}
    \liminf\limits_{n\to\infty}C_{\text{ea}}\left(\Phi_{\frac{c(N)}{n},\ketbrasub{n}},N\right)>C(\Id,N)\,.
\end{equation} 
This effectively proves Conjecture~\ref{cong0} in the regime where $\lambda\to 0^+$, establishing that environment control and entanglement assistance enable communication performance comparable to that achievable in the case in which $\lambda=1$. This phenomenon of noise neutralisation can be regarded as an EA version of the phenomenon of D-HQCOM uncovered in~\cite{die-hard}. Let us remark that the possibility of extending the above analysis beyond the infinitesimal values of $\lambda$ to the finite case is prevented by the (rather surprising) fact that typically $C_{\text{ea}}\left(\Phi_{\lambda,\sigma},N\right)$ is not monotonically increasing in $\lambda$~\cite{Die-Hard-2-PRA}.



\smallskip
\textbf{\em Control of the environment}.--- The D-HQCOM effect (uncovered in~\cite{die-hard}) and its EA 
version (uncovered in this Letter) guarantee that very good communication performances are possible even in the limit of vanishing transmissivities, under the assumption that the environment state can be suitably chosen. 
Here we remove this assumption, by providing a fully consistent protocol that exploits memory effects in order to control the environment.

Memory effects in quantum communication can be described by the collisional model formulated in~\cite{Dynamical-Model}. Translating it into the case we are analysing here, it consists in splitting the channel environment of the fibre in two components: a local term $E$ initialised into a thermal state $\tau_\nu$ that couples directly with Alice's signals via the BS interaction $U_{\lambda}^{S E}$, and a remote term $R$ which instead only interacts with $E$ trying to reset its state to $\tau_\nu$ via a thermalisation process characterised by a timescale $t_E$. {This process is described by a one-parameter family of quantum channels $\{\xi_{\delta t}\}_{\delta t\ge0}$ such that for any state $\sigma$ of $E$ it holds that (a)~$\xi_{\delta t}(\sigma)=\tau_\nu$ for $\delta t\ge t_E$ and (b) $\xi_{\delta t}(\sigma)\simeq\sigma$ for $\delta t\ll t_E$. 
We assume (a) since in this way if the time interval $\delta t$ between signals sent by Alice is such that {$\delta t\ge t_E$} the above model reduces to the best studied model of bosonic quantum communication across optical fibre, where the attenuation noise which affects each signal is represented by the same quantum channel, i.e.~by the memoryless thermal attenuator $\Phi_{\lambda,\tau_\nu}$}.
{If the time interval between signals satisfies $\delta t\ll t_E$}, the thermalisation induced by $R$ can be neglected and the dynamical evolution of $E$ will be dominated by its interactions with the transmitted signals: in this regime Alice has hence the possibility of exerting a certain level of control on the transmission line. Building up on this observation we can hence introduce a protocol that enables the communicating parties to effectively move from the memoryless channel description $\Phi_{\lambda,\tau_\nu}$ into a new effective memoryless channel $\Phi_{\lambda,\sigma}$ (see Fig.~\ref{protocol}):

\medskip
\noindent \textbf{\emph{Noise attenuation protocol}}: \vspace{-1.6ex}
\begin{itemize} \itemsep0em
\item \emph{step~1}: Alice waits for a time $t_E$ (so that the thermalisation resets $E$ into $\tau_\nu$);
\item \emph{step~2}: Alice sends $k$ suitable signals, dubbed `trigger signals', that alter $E$ into the chosen state $\sigma$;
\item \emph{step~3}: Alice sends an information-carrying signal. Then, she goes back to step~1, unless the communication is complete.
\end{itemize}

\begin{figure}[t]
\includegraphics[width=1.0\linewidth]{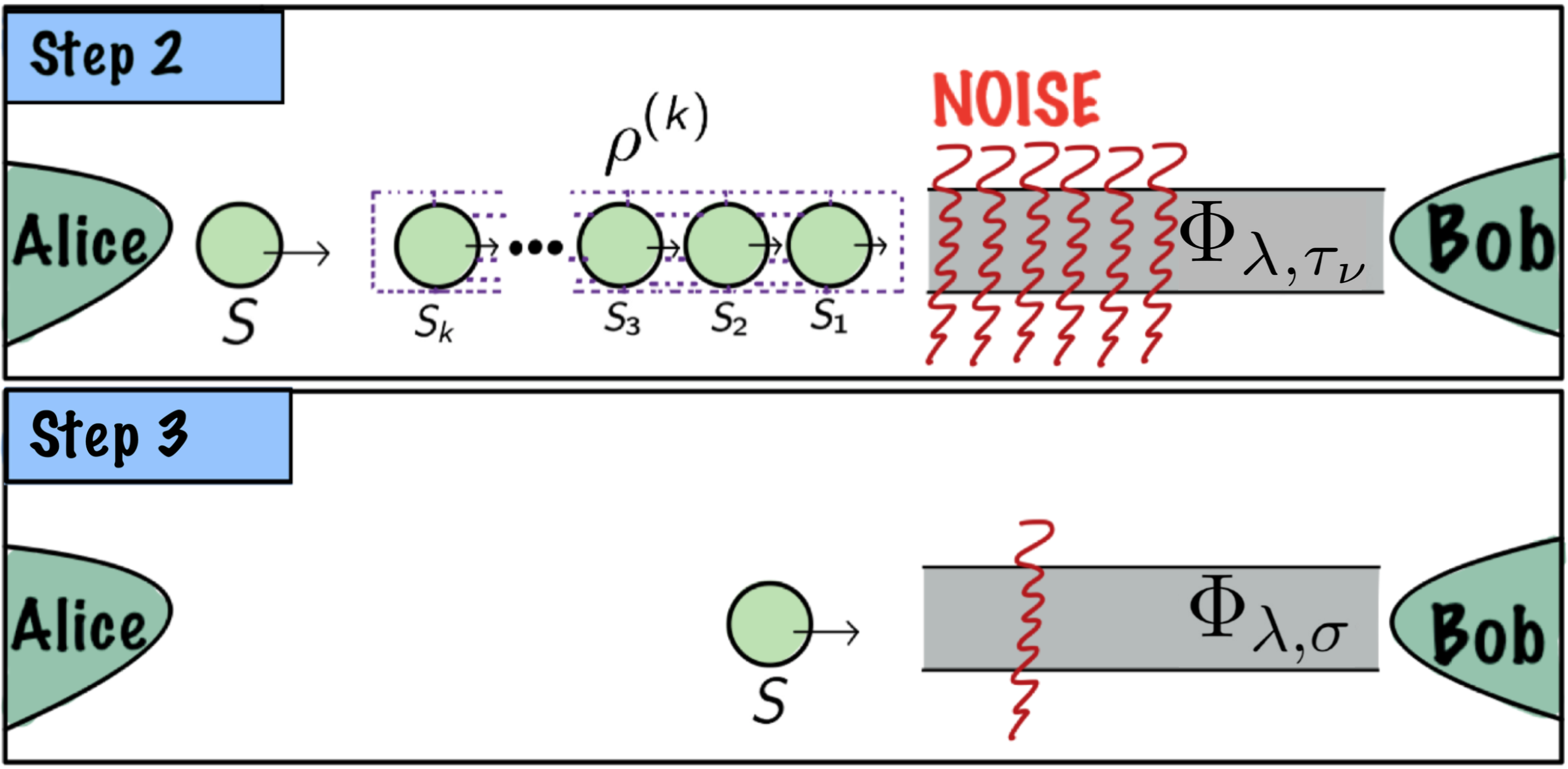}
\caption{Steps~2 and~3 of the noise attenuation protocol. At the beginning of step~2, the environment is initialised in $\tau_\nu$. By sending the signals $S_1,S_2,\ldots,S_k$, Alice aims to turn the environment into a state $\sigma$, where the latter is such that $\Phi_{\lambda,\sigma}$ is less noisy than $\Phi_{\lambda,\tau_\nu}$. Right after the environment has transformed into $\sigma$, step~3 starts with Alice sending the information-carrying signal $S$.}
\label{protocol}
\end{figure}

Let $\HH_{S_i}$ denote the Hilbert space of the $i$th trigger signal. Suppose that Alice sends $k$ trigger signals $S_1,S_2,\ldots,S_k$ separated by a time interval $\delta t$ and initialised into the state $\rho^{(k)}$. For simplicity, suppose that $\delta t$ is also the time interval between the $k$th trigger signal and the information-carrying signal $S$ of step~3. Then the state $\sigma$ of $E$, which interacts with $S$, can be expressed as
\bb \label{env_out}
\sigma = \Tr_{S_{1}\ldots S_{k}} \!\Big[\xi_{\delta t}\!\circ\! \mathcal{U}_{\lambda}^{S_kE}\!\!\ldots \circ \xi_{\delta t}\!\circ\! \mathcal{U}_{\lambda}^{S_{1}E} \Big(\rho^{(k)}\!\otimes \tau_\nu\Big) \Big],
\ee
where $\mathcal{U}_{\lambda}^{S_iE}$ is a quantum channel defined by $\mathcal{U}_{\lambda}^{S_iE}(\cdot)=U_{\lambda}^{S_iE}(\cdot)\, {U_{\lambda}^{S_iE}}^\dagger$.

From now on, suppose that $\delta t \ll t_E$ so that we can use the approximation $\xi_{\delta t}\simeq \Id$. Hence,~\eqref{env_out} reduces to
\bb \label{achievable}
\sigma=\Tr_{S_{\!1}\ldots S_k} \!\Big[U_{\lambda}^{S_kE}\!\!\ldots U_{\lambda}^{S_1E}\rho^{(k)}\!\otimes\tau_\nu \Big(U_{\lambda}^{S_kE}\!\! \ldots U_{\lambda}^{S_1E}\Big)^\dagger\Big]\,.
\ee
The phenomena of D-HQCOM can be activated if $E$ is altered in a suitable Fock state $\ketbra{n}$. Unfortunately, for all $n\in\N^+$ there does not exist $\rho^{(k)}$ that alters $E$ into $\ketbra{n}$~\cite[Theorem~13]{Die-Hard-2-PRA}. 
However, there exists $\rho^{(k)}$ that alters $E$ into a state that is as close (in trace distance) to $\ketbra{n}$ as desired if $k$ is large, as established by the following theorem~\cite{Note2}. 

\begin{thm} \label{FockAchiev}
Let $\nu\ge0$, $n\in\N$, and $\lambda\in(0,1)$. There exists a suitable state of $k$ trigger signals $\rho_{\lambda,n}^{(k)}$ such that it can alter $E$ into a state $\sigma_{\lambda,\nu,n,k}$, given by~\eqref{achievable}, which satisfies $\lim\limits_{k\to\infty}\|\sigma_{\lambda,\nu,n,k}-\ketbra{n}\|_1=0$.

Furthermore, let $\lambda\in (0,1/2)$ and set $n=n_\lambda\in\mathbb{N}$ with $1/\lambda-1\le n_\lambda\le 1/\lambda$, then for $k$ sufficiently large it holds that $Q(\Phi_{\lambda,\sigma_{\lambda,\nu,n_\lambda,k}})>0$.

In addition, if Conjecture~\ref{cong0} holds, then for all $N>0$, $\lambda\in(0,1/2)$, and for $\bar{n}\in\N$ sufficiently large, it holds that
\bb
    \lim\limits_{k\to\infty}C_{\text{ea}}\left(\Phi_{\sigma_{\lambda,\nu,\bar{n},k}},N\right)&>C\left(\Id,N\right)\,,\\
    \lim\limits_{k\to\infty}Q_{\text{ea}}\left(\Phi_{\lambda,\sigma_{\lambda,\nu,\bar{n},k}},N\right)&>Q\left(\Id,N\right)/2\,,\\
    \lim\limits_{k\to\infty}Q\left(\Phi_{\lambda,\sigma_{\lambda,\nu,\bar{n},k}},N\right)&>0\,.
\ee
\end{thm}
Theorem~\ref{FockAchiev} implies that the D-HQCOM effect can be activated by applying the noise attenuation protocol with a sufficiently large number $k$ of trigger signals initialised in the state $\rho_{\lambda,n}^{(k)}$: an explicit construction to produce 
such pulses using Fock states and linear optics can be found in~\cite{Die-Hard-2-PRA}. Applying this protocol 
achieves a dramatic improvement of the communication performance of an optical fibre. Indeed, if an optical fibre with transmissivity $0<\lambda\ll 1/2$ is used as usual --- i.e.\ by sending signals separated by a time interval $\delta t \gtrsim t_E$ --- then it is described by a thermal attenuator which has zero quantum capacity and vanishing (two-way- or entanglement-) assisted capacities. 
A drawback of this construction, however, is that, counting the transmission of the trigger signals as channel uses, the rate it achieves is equal to $Q(\Phi_{\lambda,\sigma})/(k+1)$, which can be small if the construction in Theorem~\ref{FockAchiev} yields a large $k$. 
Fortunately, for $\lambda>0$ sufficiently small, just $k=2$ is enough to guarantee non-zero quantum capacity~\cite{Note2}.
\begin{thm}\label{theorem_lambda_small}
Let $\nu\ge0$. Suppose that Alice sends two trigger signals in $U_{\frac{1}{1+\lambda}}^{S_1S_2}\ket{0}_{S_1}\ket{n_\lambda}_{S_2}$,
	with $n_\lambda\in\mathbb{N}$ such that $1/\lambda-1\le n_\lambda\le 1/\lambda$.
	Then, $E$ is altered into a state $\sigma_{\lambda,\nu}$ such that 
	for $\lambda>0$ sufficiently small it holds that $Q\left(\Phi_{\lambda,\sigma_{\lambda,\nu}}\right)\ge Q\left(\Phi_{\lambda,\sigma_{\lambda,\nu}},1/2\right)\ge c$,
	where $c>0$ is a fixed positive constant. 
\end{thm}
Theorem~\ref{theorem_lambda_small} shows that it is possible to reliably transmit qubits at a fixed positive rate if $\lambda>0$ is sufficiently low, i.e.\ if the optical fibre is sufficiently long. Theorem~\ref{FockAchiev} and Theorem~\ref{theorem_lambda_small} are valid not only if the equilibrium state of the thermalisation process is a thermal state $\tau_\nu$, but also if it is any state $\sigma_0$ such that $\langle (b^\dagger b)^2\rangle_{\sigma_0}<\infty$~\cite{Die-Hard-2-PRA}.

The analysis presented in this section is based on the expression of the state of $E$ in~\eqref{achievable}, which is an approximation valid for $\delta t\ll t_E$. A similar analysis can also be 
carried out using the exact expression in~\eqref{env_out}. The latter reduces to~\eqref{achievable} if $\xi_{\delta t}=\Id$. It turns out that, under suitable continuity properties of $\xi_{\delta t}$, if $\rho^{(k)}$ is such that the general attenuator with environment state in~\eqref{achievable} has strictly positive quantum capacity, then $\rho^{(k)}$ is such that the general attenuator with environment state in~\eqref{env_out} has strictly positive quantum capacity for $\delta t$ sufficiently short~\cite[Theorem 19]{Die-Hard-2-PRA}. In other words, any scheme for quantum communication that works with the approximation $\xi_{\delta t}=\Id$ presented in this section will also work without such an approximation for $\delta t$ sufficiently short.

\smallskip
\textbf{\em Conclusions.}--- This Letter shows how memory effects can be engineered in order to improve the communication performance of an optical fibre. This is done by exploiting the noise attenuation protocol with trigger signals initialised in a suitable state. 
This protocol allows arbitrarily long optical fibres to reliably transmit: (a)~qubits at a fixed positive rate; (b)~bits and qubits at a rate of the same order of the maximum achievable in the ideal case of absence of noise, provided that pre-shared entanglement is consumed. An interesting development of our analysis is to take into account the decoherence process suffered by the signals as they wait to be fed into the optical fibre, due to imperfections in Alice's quantum memory, e.g.\ by exploiting the recently proposed queuing framework~\cite{queue1, queue2, queue3}.

We encourage experimental research on memory effects in optical fibres. For instance, one could test the model we have presented and estimate the thermalisation time $t_E$. Since the latter may be not much larger than the shortest time interval between subsequent signals allowed by nowadays devices, an interesting development of this work is to analyse the noise attenuation protocol without our simplifying hypothesis $\delta t\ll t_E$. 


\smallskip
\textbf{\em Acknowledgements.}--- FAM and VG acknowledge financial support by MIUR (Ministero dell'Istruzione, dell'Universit\`a e della Ricerca) via project PRIN 2017 ``Taming complexity via Quantum Strategies: a Hybrid Integrated Photonic approach'' (QUSHIP) Id.\ 2017SRNBRK, and via project PRO3 ``Quantum Pathfinder''. LL acknowledges financial support from the Alexander von Humboldt Foundation.

\bibliographystyle{unsrt}
\bibliography{biblio}

\newpage
\clearpage

\onecolumngrid
\begin{center}
\vspace*{\baselineskip}
{\textbf{\large Supplemental material:\\ Restoring quantum communication efficiency over high loss optical fibres}}\\
\end{center}

\renewcommand{\theequation}{S\arabic{equation}}
\renewcommand{\thethm}{S\arabic{thm}}
\setcounter{equation}{0}
\setcounter{thm}{0}
\setcounter{figure}{0}
\setcounter{table}{0}
\setcounter{section}{0}
\setcounter{page}{1}
\makeatletter

\setcounter{secnumdepth}{2}

\section{Environment control and entanglement-assistance imply noise neutralisation}
\begin{Def}
Fixed $\lambda\in[0,1]$ and $\sigma\in\mathfrak{S}(\HH_E)$, a general attenuator $\Phi_{\lambda,\sigma}:\mathfrak{S}(\HH_S)\mapsto\mathfrak{S}(\HH_S)$ is a quantum channel defined by 
	\begin{equation}\label{def_genatt}
	\Phi_{\lambda,\sigma}(\rho)\coloneqq\Tr_E\left[U_\lambda^{(SE)}  \rho\otimes\sigma  \left(U_\lambda^{(SE)}\right)^\dagger\right]\,,
	\end{equation}
where $U_{\lambda}^{(SE)}$ is the beam splitter (BS) unitary defined by 
	\bb
	U_{\lambda}^{(SE)}\coloneqq e^{\arccos{\sqrt{\lambda}}\left(a^\dagger b-ab^\dagger\right)}\,,
	\ee
	with $a$ and $b$ being the annihilation operators of the system $S$ and of the environment $E$, respectively.
The weak complementary channel of the general attenuator $\Phi_{\lambda,\sigma}$ is defined by
\begin{equation}
	\Phi_{\lambda,\sigma}^{\text{wc}}(\rho)\coloneqq\Tr_S\left[U_\lambda^{(SE)}  \rho\otimes\sigma  \left(U_\lambda^{(SE)}\right)^\dagger\right]\,.
	\end{equation}
\end{Def}

\begin{lemma}\label{lemmatrasf}
It holds that:
	\begin{equation}\label{trasfheisa}
	\left(U_{\lambda}^{(SE)}\right)^\dagger a\,U_{\lambda}^{(SE)}=\sqrt{\lambda}a+\sqrt{1-\lambda}b\,;
	\end{equation}
	\begin{equation}\label{trasfheisb}
	\left(U_{\lambda}^{(SE)}\right)^\dagger b\,U_{\lambda}^{(SE)}=-\sqrt{1-\lambda}a+\sqrt{\lambda}b\,;
	\end{equation}
	\begin{equation}\label{trasfheisa2}
	{U_{\lambda}^{(SE)}} a\left(U_{\lambda}^{(SE)}\right)^\dagger=\sqrt{\lambda}a-\sqrt{1-\lambda}b\,;
	\end{equation}
	\begin{equation}\label{trasfheisb2}
	{U_{\lambda}^{(SE)}} b\left(U_{\lambda}^{(SE)}\right)^\dagger=\sqrt{1-\lambda}a+\sqrt{\lambda}b\,.
	\end{equation}
\end{lemma}
\begin{proof}
	By setting $\eta\coloneqq \arccos\sqrt{\lambda}$ i.e.~$\lambda(\eta)=\cos^2(\eta
	)$, and $$\hat{f}(\eta)\coloneqq	\left(U_{\lambda(\eta)}^{(SE)}\right)^\dagger a\,U_{\lambda(\eta)}^{(SE)}=e^{-\eta\left(a^\dagger b-ab^\dagger\right)}a\,e^{\eta\left(a^\dagger b-ab^\dagger\right)}\,,$$
	one obtains:
	\bb\label{f'}
	\hat{f}'(\eta)&=-\left(U_{\lambda(\eta)}^{(SE)}\right)^\dagger [a^\dagger b-ab^\dagger,a]\,U_{\lambda(\eta)}^{(SE)}=	\left(U_{\lambda(\eta)}^{(SE)}\right)^\dagger b\,U_{\lambda(\eta)}^{(SE)}\,,
	\ee
	\bb\label{f'}
	\hat{f}'(\eta)&=-\left(U_{\lambda(\eta)}^{(SE)}\right)^\dagger [a^\dagger b-ab^\dagger,b]\,U_{\lambda(\eta)}^{(SE)}=	-\left(U_{\lambda(\eta)}^{(SE)}\right)^\dagger a\,U_{\lambda(\eta)}^{(SE)}=-\hat{f}(\eta)\,.
	\ee
	Therefore, there exist two operators $\hat{c_0}$ and $\hat{c_1}$ such that
	\begin{equation}\label{formfeta}
	\hat{f}(\eta)=\cos\eta \,\hat{c_0}+\sin\eta \,\hat{c_1}\,.
	\end{equation}
	By imposing $\hat{f}(0)=a$ and $\hat{f}'(0)=b$, we arrive at
	$$\hat{f}(\eta)=\cos\eta \,a+\sin\eta\, b$$
	which implies \eqref{trasfheisa} to be valid. Furthermore, \eqref{trasfheisb} follows from~\eqref{f'}. To conclude the proof, \eqref{trasfheisa2} follows from~\eqref{trasfheisa} by substituting $b\rightarrow-b$, while \eqref{trasfheisb2} follows from~\eqref{trasfheisb} by substituting $a\rightarrow-a$.
\end{proof}
\begin{lemma}\label{LemmaEner}
The mean photon number of the output system and environment are given by
\begin{equation}\label{enphi}
	\langle a^\dagger a\rangle_{\Phi_{\lambda,\sigma}(\rho)} \coloneqq \Tr_S\left[\Phi_{\lambda,\sigma}(\rho)\,a^\dagger a\right]= \lambda\langle a^\dagger a\rangle_\rho +(1-\lambda)\langle b^\dagger b\rangle_\sigma+2\sqrt{\lambda(1-\lambda)}\Re\left(\langle a\rangle_\rho \langle b^\dagger\rangle_\sigma\right)
\end{equation}
and
\begin{equation}\label{enweak}
\langle b^\dagger b\rangle_{\tilde{\Phi}^{\text{wc}}_{\lambda,\sigma}(\rho)} =(1-\lambda)\langle a^\dagger a\rangle_\rho +\lambda\langle b^\dagger b\rangle_\sigma-2\sqrt{\lambda(1-\lambda)}\Re\left(\langle a\rangle_\rho \langle b^\dagger\rangle_\sigma\right)\,,
\end{equation}
respectively.
\end{lemma}
\begin{proof}
From Lemma~\ref{lemmatrasf} we have
\bb
\langle a^\dagger a\rangle_{\Phi_{\lambda,\sigma}(\rho)} &=	\Tr_S\left[\Phi_{\lambda,\sigma}(\rho)\,a^\dagger a\right]=\Tr_{SE}\left[U_\lambda^{(SE)} \rho\otimes\sigma \left({U_\lambda^{(SE)}}\right)^\dagger a^\dagger a\right]\\&=\Tr_{SE}\left[\rho\otimes\sigma \left(\left({U_\lambda^{(SE)}}\right)^\dagger a^\dagger U_\lambda^{(SE)}\right)\left(\left({U_\lambda^{(SE)}}\right)^\dagger a U_\lambda^{(SE)}\right)\right]\\&=\Tr_{SE}\left[\rho\otimes\sigma \left(\sqrt{\lambda}a^\dagger+\sqrt{1-\lambda}b^\dagger\right)\left(\sqrt{\lambda}a+\sqrt{1-\lambda}b\right)\right]\\&=\lambda\langle a^\dagger a\rangle_\rho +(1-\lambda)\langle b^\dagger b\rangle_\sigma+\sqrt{\lambda(1-\lambda)}\left(\langle a\rangle_\rho \langle b^\dagger\rangle_\sigma+\langle a^\dagger\rangle_\rho \langle b\rangle_\sigma\right)\\&=\lambda\langle a^\dagger a\rangle_\rho +(1-\lambda)\langle b^\dagger b\rangle_\sigma+2\sqrt{\lambda(1-\lambda)}\Re\left(\langle a\rangle_\rho \langle b^\dagger\rangle_\sigma\right)\,.
\ee
The formula~\eqref{enweak} follows from a similar calculation. 
\end{proof}
\begin{Def}
For every trace class operator $T$ on $L^2(\mathbb{R}^m)$, its \emph{characteristic function} $\chi_{T}:\mathbb{C^m}\rightarrow\mathbb{C}$ is defined by
\begin{equation}
	\chi_{T}(z)\coloneqq \Tr\left[T\,D(z)\right]\qquad\forall\, z=(z_1,z_2,\ldots,z_m)\in\mathbb{C}^m \,,
\end{equation}
	where $$D(z)\coloneqq\exp{\left[\sum_{i=0}^m \left(z_i a_i^\dagger-z_i^\ast a_i\right)\right] }$$ is the displacement operator, with $a_i$ being the annihilation operator corresponding to the i-th mode.
Conversely, it turns out that every trace class operator $T$ can be reconstructed from its characteristic functions $\chi_{T}$ via the following identity~\cite{HOLEVO-CHANNELS-2,BUCCO}:
\begin{equation}
T=\int_{\mathbb{C}} \frac{d^m z}{\pi^m}\text{ }D(-z)\chi_{T}(z)\,.
\end{equation}
\end{Def}

\begin{lemma}
The action of a {BS} of transmissivity $\lambda$ on the a two-mode state of the form $\rho\otimes\sigma$ can be cast in the language of characteristic functions as
\begin{equation}\label{caract}
\chi_{U_\lambda^{(SE)}  \rho\otimes\sigma  \left(U_\lambda^{(SE)}\right)^\dagger} \left(z,w\right)=\chi_{\rho}\left(\sqrt{\lambda}z-\sqrt{1-\lambda}w\right)\chi_{\sigma}\left(\sqrt{1-\lambda}z+\sqrt{\lambda}w\right)\quad\forall\, z,w\in\mathbb{C}\,.
\end{equation}
Consequently, it holds that
\begin{equation}\label{caract3}
\chi_{\Phi_{\lambda,\sigma}(\rho)} \left(z\right)=\chi_{\rho}\left(\sqrt{\lambda}z\right)\chi_{\sigma}\left(\sqrt{1-\lambda}z\right)\quad\forall\, z\in\mathbb{C}\,,
\end{equation}
\begin{equation}\label{caract_weak}
\chi_{\tilde{\Phi}^{\text{wc}}_{\lambda,\sigma}(\rho)} \left(z\right)=\chi_{\rho}\left(-\sqrt{1-\lambda}w\right)\chi_{\sigma}\left(\sqrt{\lambda}w\right)\quad\forall\, w\in\mathbb{C}\,.
\end{equation}
\end{lemma}
\begin{proof}
Denote the displacement operators on $\HH_S$ and $\HH_E$ respectively as $D_S(z)$ and $D_E(z)$.
\eqref{trasfheisa} and \eqref{trasfheisb} imply that
\bb
 \left(U_\lambda^{(SE)}\right)^\dagger D_S(z)\, U_\lambda^{(SE)}&=D_S(\sqrt{\lambda}z)\, D_E(\sqrt{1-\lambda}z)\\
 \left(U_\lambda^{(SE)}\right)^\dagger D_E(w)\, U_\lambda^{(SE)}&=D_S(-\sqrt{1-\lambda}w)\, D_E(\sqrt{\lambda}w)
\ee
Consequently, we obtain
\begin{equation}
 \left(U_\lambda^{(SE)}\right)^\dagger D_S(z)\, D_E(w)\, U_\lambda^{(SE)}=D_S(\sqrt{\lambda}z-\sqrt{1-\lambda}w)\, D_E(\sqrt{1-\lambda}z+\sqrt{\lambda}w)\,.
\end{equation}
Hence, it holds that
\bb
\chi_{U_\lambda^{(SE)}  \rho\otimes\sigma  \left(U_\lambda^{(SE)}\right)^\dagger} \left(z,w\right)&=\Tr_{SE}\left[U_\lambda^{(SE)}  \rho\otimes\sigma  \left(U_\lambda^{(SE)}\right)^\dagger \,D_S(z)\, D_E(z) \right] \\&=\Tr_{SE}\left[\rho\otimes\sigma\,D_S(\sqrt{\lambda}z-\sqrt{1-\lambda}w)\, D_E(\sqrt{1-\lambda}z+\sqrt{\lambda}w) \right] \\&=\chi_{\rho}\left(\sqrt{\lambda}z-\sqrt{1-\lambda}w\right)\chi_{\sigma}\left(\sqrt{1-\lambda}z+\sqrt{\lambda}w\right)\,.
\ee
\eqref{caract3} and \eqref{caract_weak} follow from the fact that for all $\rho_{SE}\in\mathfrak{S}(\HH_S\otimes\HH_E)$ it holds that
\bb
    \chi_{\Tr_S\rho_{SE}}(w)=\chi_{\rho_{SE}}(z=0,w)\,,
    \chi_{\Tr_E\rho_{SE}}(z)=\chi_{\rho_{SE}}(z,w=0)\,.
\ee
\end{proof}

\begin{lemma}\label{lemma_invert}
For all $\lambda\in[0,1]$ and all $\rho$, $\sigma$ single-mode states, it holds that
\begin{equation}\label{invert2}
\Phi_{\lambda,\sigma}(\rho)=\Phi_{1-\lambda,\,\rho}(\sigma)\,.
\end{equation}
\end{lemma}
\begin{proof}
The characteristic function associated with an output of a general attenuator can be expressed as in \eqref{caract3} and hence it holds that
\begin{equation}
	\chi_{\Phi_{\lambda,\sigma}(\rho)} \left(z\right)=\chi_{\Phi_{1-\lambda,\rho}(\sigma)} \left(z\right)\quad\forall\, z\in\mathbb{C}\,.
\end{equation}
Since quantum states are in one-to-one correspondence with characteristic functions, \eqref{invert2} is proved.
\end{proof}

In the proof of Lemma~\ref{teo_lowtrasm} we will use the following result, dubbed `master equation trick', which is derived in \cite{Die-Hard-2-PRA} by exploiting the master equation associated with the thermal attenuator $\Phi_{\lambda,\tau_\nu}$.
\begin{thm}{\cite[Theorem 5]{Die-Hard-2-PRA}}\label{master_trick}
For all $\lambda\in[0,1],\nu\ge 0$, the thermal attenuator $\Phi_{\lambda,\tau_\nu}$ admits the following Kraus representation:
\begin{equation}\label{krausformatt}
    \Phi_{\lambda,\tau_\nu}(\rho)=\sum_{k,m=0}^\infty M_{k,m}\rho M_{k,m}^\dagger\,,
\end{equation}
where 
\bb
    M_{k,m}\coloneqq\sqrt{\frac{\nu^k(\nu+1)^m(1-\lambda)^{m+k}}{k!m![(1-\lambda)\nu+1]^{m+k+1}}}(a^\dagger)^k\left(\frac{\sqrt{\lambda}}{(1-\lambda)\nu+1}\right)^{a^\dagger a}a^m\,.
\ee
In particular, by letting $\ket{n}$ and $\ket{i}$ two Fock states, it holds that
\begin{equation}\label{action_ni}
    \Phi_{\lambda,\tau_\nu}(\ketbraa{n}{i})=\sum_{l=\max(i-n,0)}^\infty f_{n,i,l}(\lambda,\nu)\ketbraa{l+n-i}{l}\,,
\end{equation}
where
\bb\label{f_nil}
    f_{n,i,l}(\lambda,\nu)\coloneqq \sum_{m=\max(i-l,0)}^{\min(n,i)}\frac{\sqrt{n!i!l!(l+n-i)!}}{(n-m)!(i-m)!m!(l+m-i)!}\frac{\nu^{l+m-i}(\nu+1)^m(1-\lambda)^{2m+l-i}\lambda^{\frac{n+i-2m}{2}}}{\left((1-\lambda)\nu+1\right)^{l+n+1}}\,.
\ee
\end{thm}
Let us state a useful lemma which follows from the proof of~\cite[Theorem 1]{die-hard}.
\begin{lemma}\cite[Proof of Theorem 1]{die-hard}\label{lemma_diehard}
Let $\alpha\in\mathbb{C}$. Let $\sigma\in\mathfrak{S}(\HH_E)$ be of the form $\sigma=D(\alpha)\sigma_0{D(\alpha)}^\dagger$, where:
\begin{equation}
    D(z)\coloneqq\exp\left[z\, b^\dagger-z^\ast\, b\right]
\end{equation}
is the displacement operator, and $\sigma_0$ satisfies 
\begin{equation}
    V \sigma_0 V^\dagger=\sigma_0\,,
\end{equation}
with $V \coloneqq (-1)^{b^\dagger b}$ being the parity operator. Then, it holds that
\begin{equation}\label{weak_formula}
\tilde{\Phi}_{\lambda,\sigma}^{\text{wc}}=\mathcal{V}\circ\mathcal{D}_{-2\sqrt{\lambda}\alpha}\circ\Phi_{1-\lambda,\sigma}\,,
\end{equation}
where $\mathcal{V}$ and $\mathcal{D}_{\alpha}$ are two quantum channels defined by: 
\begin{equation}\label{tV}
\mathcal{V}(\cdot)\coloneqq V (\cdot)  V^\dagger\,,
\end{equation}
\begin{equation}\label{tD}
\mathcal{D}_\alpha(\cdot)\coloneqq D(\alpha) (\cdot) D(\alpha)^\dagger\,.
\end{equation}
\end{lemma}

\begin{lemma}
For all $\lambda\in[0,1]$, $N>0$, and $n\in\N$, it holds that
\bb \label{cohprob}
	&I_{\text{coh}}\left(\Phi_{\lambda,\ketbrasub{n}},\tau_N\right)  =H\left(P(N,n,\lambda)\right)-H\left(P(N,n,1-\lambda)\right)\, ,
\ee
where
\bb\label{simpler2}
    &P_l(N,n,\lambda)
    =\frac{1}{\left(1+N\lambda\right)^{l+n+1}}\sum_{m=\max\{0,n-l\}}^n\lambda^{2m+l-n}(1-\lambda)^{n-m}  N^{l+m-n}(N+1)^m\binom{n}{m}\binom{l}{n-m}\,.
\ee
$H(P)$ denotes the Shannon entropy of the probability distribution $P$ i.e.\ $H(P)\coloneqq -\sum_{i: P_i\ne 0}P_i\log_2P_i$.
\end{lemma} 
\begin{proof}
The coherent information is equal to~\cite{MARK,HOLEVO-CHANNELS-2}
\bb
&I_{\text{coh}}\left(\Phi_{\lambda,\ketbrasub{n}},\tau_N\right)  =S\left(\Phi_{\lambda,\ketbrasub{n}}(\tau_N)\right)-S\left(\tilde{\Phi}_{\lambda,\ketbrasub{n}}(\tau_N)\right)\,,
\ee
where $\tilde{\Phi}_{\lambda,\sigma}$ denotes a complementary channel of ${\Phi}_{\lambda,\ketbrasub{n}}$.
Since the weak complementary channel $\tilde{\Phi}^{\text{wc}}_{\lambda,\ketbrasub{n}}$ is associated with a pure environment state, it is a complementary channel of ${\Phi}_{\lambda,\ketbrasub{n}}$. Hence, it holds that
\bb
&I_{\text{coh}}\left(\Phi_{\lambda,\ketbrasub{n}},\tau_N\right)  =S\left(\Phi_{\lambda,\ketbrasub{n}}(\tau_N)\right)-S\left(\tilde{\Phi}^{\text{wc}}_{\lambda,\ketbrasub{n}}(\tau_N)\right)\,.
\ee
By applying Lemma~\ref{lemma_diehard} and the invariance of the von Neumman entropy under unitary transformation, we obtain
\begin{equation}
    S\left(\tilde{\Phi}^{\text{wc}}_{\lambda,\ketbrasub{n}}(\tau_N)\right)= S\left(\Phi_{1-\lambda,\ketbrasub{n}}(\tau_N)\right)\,.
\end{equation}
As a consequence,
\bb\label{Icoh}
&I_{\text{coh}}\left(\Phi_{\lambda,\ketbrasub{n}},\tau_N\right) =S\left(\Phi_{\lambda,\ketbrasub{n}}(\tau_N)\right) -S\left(\Phi_{1-\lambda,\ketbrasub{n}}(\tau_N)\right)\,.
\ee 
Since Lemma~\ref{lemma_invert} implies that
\begin{equation}\label{scambio}
\Phi_{\lambda,\ketbrasub{n}}(\tau_N)=\Phi_{1-\lambda,\tau_N}(\ketbra{n})\,,
\end{equation}
we obtain
\bb\label{Icoh22}
&I_{\text{coh}}\left(\Phi_{\lambda,\ketbrasub{n}},\tau_N\right) =S\left(\Phi_{1-\lambda,\tau_N}(\ketbrasub{n})\right) -S\left(\Phi_{\lambda,\tau_N}(\ketbrasub{n})\right)\,.
\ee 
By using Theorem~\ref{entdifflow}, we can write
\bb\label{phia_NEW}
\Phi_{\lambda,\tau_N}(\ketbraa{n}{n})=\sum_{l=0}^\infty P_l(N,n,1-\lambda)\ketbraa{l}{l}\,,
\ee
where we have denoted $P_l(N,n,\lambda)\coloneqq f_{n,n,l}(1-\lambda,N)$. By exploiting \eqref{f_nil}, one can show that this definition of $P_l(N,n,\lambda)$ leads to \eqref{simpler2}.  Since the von Neumann entropy of $\Phi_{\lambda,\tau_N}(\ketbrasub{n})$ is the Shannon entropy of the probability distribution $\{P_l(N,n,1-\lambda)\}_{l\in\N}$, \eqref{Icoh22} implies that $I_{\text{coh}}\left(\Phi_{\lambda,\ketbrasub{n}},\tau_N\right)  =H\left(P(N,n,\lambda)\right)-H\left(P(N,n,1-\lambda)\right)$.
\end{proof}

\begin{lemma*}[\ref{teo_lowtrasm}]
For all $N,c>0$ it holds that
	\begin{align}\label{entdifflow}
		\liminf\limits_{n\to\infty}I_{\mathrm{coh}}\left(\Phi_{\frac{c}{n},\ketbrasub{n}},\tau_N\right) \ge H\left(q(N,c)\right)-H\left(p(N,c)\right)\,,
	\end{align}
	where $\{q_{k}(N,c)\}_{k\in\mathbb{Z}}$ and $\{p_{k}(N,c)\}_{k\in\mathbb{Z}}$ are two probability distributions defined as
\begin{equation*}
    \begin{aligned}
    q_k(N,c)&\coloneqq e^{-c(2N+1)}\left(\frac{N}{N+1}\right)^{k/2}I_{|k|}\left(2c\sqrt{N(N+1)}\right)\,,\\
	p_k(N,c)&\coloneqq
	\begin{cases}
	\frac{e^{-c/(N+1)}N^k}{(N+1)^{k+1}}L_k\left(-\frac{c}{N(N+1)}\right) &\text{if $k\ge 0$ ,} \\
	0 &\text{otherwise.}
	\end{cases}
    \end{aligned}
\end{equation*}
	with $I_k(\cdot)$ and $L_k(\cdot)$ being the $k$th Bessel function of the first kind and the $k$th Laguerre polynomial, respectively:
	\begin{equation}\label{besself}
	I_k(z)\coloneqq\left(\frac{z}{2}\right)^k 	\sum_{m=0}^{\infty}\frac{\left(\frac{z}{2}\right)^{2m}}{m!(k+m)!}\,,
	\end{equation}
	\begin{equation}
	L_k(x)\coloneqq\sum_{m=0}^k\frac{(-1)^m}{m!}\binom{k}{m}x^m\,.
	\end{equation}
\end{lemma*}

\begin{proof}
From \eqref{Icoh}, we have  
\bb\label{conscons}
&I_{\text{coh}}\left(\Phi_{\frac{c}{n},\ketbrasub{n}},\tau_N\right)=S\left(\Phi_{\frac{c}{n},\ketbrasub{n}}(\tau_N)\right) -S\left(\Phi_{1-\frac{c}{n},\ketbrasub{n}}(\tau_N)\right)\,.
\ee
By using \eqref{scambio} and \eqref{phia_NEW}, we can write
\bb\label{phia}
\Phi_{\frac{c}{n},\ketbrasub{n}}(\tau_N)&=\sum_{l=0}^{\infty}P_l\left(N,n,\frac{c}{n}\right)\ketbra{l}\,,\\
\Phi_{1-\frac{c}{n},\ketbrasub{n}}(\tau_N)&=\sum_{l=0}^{\infty}P_l\left(N,n,1-\frac{c}{n}\right)\ketbra{l}\,,\
\ee
where the coefficients $P_l(N,n,\lambda)$ are equal to
\bb\label{simpler}
    &P_l(N,n,\lambda)
    =\frac{1}{\left(1+N\lambda\right)^{l+n+1}}\sum_{m=\max\{0,n-l\}}^n\lambda^{2m+l-n}(1-\lambda)^{n-m}  N^{l+m-n}(N+1)^m\binom{n}{m}\binom{l}{n-m}\,.
\ee
Let us extend our harmonic oscillator Hilbert space by adding new orthonormal states $\{\ket{-i}\}_{i\in\N^+}$ so that $\{\ket{k}\}_{k\in\mathbb{Z}}$ is an orthonormal basis of the resulting Hilbert space $\HH_{\text{ext}}$.
We define the extension of the operators $\Phi_{\frac{c}{n},\ketbrasub{n}}(\tau_N)$ and $\Phi_{1-\frac{c}{n},\ketbrasub{n}}(\tau_N)$ to $\HH_{\text{ext}}$ such that for all $k\in\N^+$ it holds that (for the sake of simplicity we are going to denote an operator and its extension in the same way) 
\bb
\Phi_{\frac{c}{n},\ketbrasub{n}}(\tau_N)\ket{-k}&=0\,,\\
\Phi_{1-\frac{c}{n},\ketbrasub{n}}(\tau_N)\ket{-k}&=0\,.
\ee
Moreover, let us define for all $k\in\mathbb{Z}$ the shift operators $T_k$ which act on the extended Hilbert space as
\begin{equation}
T_{k}\ket{i}=\ket{i+k}\quad\forall\,  i\in\mathbb{Z}\,.
\end{equation}
Since $T_k$ is unitary, it holds that
\begin{equation}
S\left(\Phi_{\frac{c}{n},\ketbrasub{n}}(\tau_N)\right) =S\left(T_{-n}\Phi_{\frac{c}{n},\ketbrasub{n}}(\tau_N)T_{-n}^\dagger\right).
\end{equation}
Consequently, \eqref{conscons} implies that
\bb\label{conscons2}
&I_{\text{coh}}\left(\Phi_{\frac{c}{n},\ketbrasub{n}},\tau_N\right) =S\left(T_{-n}\Phi_{\frac{c}{n},\ketbrasub{n}}(\tau_N)T_{-n}^\dagger\right)  -S\left(\Phi_{1-\frac{c}{n},\ketbrasub{n}}(\tau_N)\right)\,.
\ee
Unlike the sequence of operators 
\bb
\left\{\Phi_{\frac{c}{n},\ketbrasub{n}}(\tau_N)\right\}_{n\in\N}\subseteq\mathfrak{S}(\HH_S)\,,
\ee
its shifted version, i.e.
\bb
\left\{T_{-n}\Phi_{\frac{c}{n},\ketbrasub{n}}(\tau_N)T_{-n}^\dagger\right\}_{n\in\N}\subseteq\mathfrak{S}(\HH_{\text{ext}})\,,
\ee
does converge in trace norm as $n\rightarrow\infty$, as we are going to prove now.

\eqref{phia} implies that
\bb\label{shiftedState}
T_{-n}\Phi_{\frac{c}{n},\ketbrasub{n}}(\tau_N)T_{-n}^\dagger&=\sum_{l=0}^{\infty}P_l\left(N,n,\frac{c}{n}\right)\ketbra{l-n} =\sum_{k=-n}^{\infty}P_{n+k}\left(N,n,\frac{c}{n}\right)\ketbra{k}\,.
\ee
Hence, since we want to study $T_{-n}\Phi_{\frac{c}{n},\ketbrasub{n}}(\tau_N)T_{-n}^\dagger$ in the limit $n\rightarrow\infty$, let us define the probability distribution $\{q_k(N,c)\}_{k\in\mathbb{Z}}$ over the alphabet $\mathbb{Z}$ as
\begin{equation}\label{qk}
q_k(N,c)\coloneqq\lim\limits_{n\rightarrow\infty}P_{n+k}\left(N,n,\frac{c}{n}\right)\,.
\end{equation}
We will show below that such limit exists. To this end, let us first consider the case where $k\geq 0$. Then \eqref{simpler} implies that
\bb
&q_{-k}(N,c) \\
&= \lim\limits_{n\rightarrow\infty}\frac{\sum_{m=k}^n\left(\frac{c}{n}\right)^{2m-k}\!\left(1\!-\!\frac{c}{n}\right)^{n-m} \! N^{m-k}(N\!+\!1)^m\binom{n}{m}\!\binom{n-k}{n-m}}{(1+N\frac{c}{n})^{2n-k+1}} \\
&=\lim\limits_{n\rightarrow\infty}\frac{1}{(1+N\frac{c}{n})^{2n-k+1}}\sum_{m=0}^{n-k}\left(\frac{c}{n}\right)^{2m+k}\left(1-\frac{c}{n}\right)^{n-m-k} N^{m}(N+1)^{m+k}\binom{n}{m+k}\binom{n-k}{m}  \\&=\lim\limits_{n\rightarrow\infty}\frac{\left(1-\frac{c}{n}\right)^{n-k}\left[c(N+1)\right] ^k}{(1+N\frac{c}{n})^{2n-k+1}} \sum_{m=0}^{n-k}\left(1-\frac{c}{n}\right)^{-m}\left[c\sqrt{N(N+1)}\right]^{2m}\frac{\binom{n}{m+k}\binom{n-k}{m}}{n^{2m+k}} \\
&=e^{-c\left[2N+1\right]}\left[c(N+1)\right]^{k} \lim\limits_{n\rightarrow\infty}\sum_{m=0}^{n-k}\left(1-\frac{c}{n}\right)^{-m}\left[c\sqrt{N(N+1)}\right]^{2m}\frac{\binom{n}{m+k}\binom{n-k}{m}}{n^{2m+k}}\,,\label{seriess}
\ee
where in the last line we used the celebrated \emph{Nepero limit} i.e.$\lim\limits_{n\rightarrow\infty}\left(1-\frac{c}{n}\right)^{n}=e^{-c} $.
Now, we are going to invoke \emph{Tannery's theorem} in order to show that we can interchange the limit and the series, i.e.
\bb
	&\lim\limits_{n\rightarrow\infty}\sum_{m=0}^{n-k}\left(1-\frac{c}{n}\right)^{-m}\left[c\sqrt{N(N+1)}\right]^{2m}\frac{\binom{n}{m+k}\binom{n-k}{m}}{n^{2m+k}} \\&=\sum_{m=0}^{\infty}\lim\limits_{n\rightarrow\infty}\left(1-\frac{c}{n}\right)^{-m}\left[c\sqrt{N(N+1)}\right]^{2m}\frac{\binom{n}{m+k}\binom{n-k}{m}}{n^{2m+k}}\,.
\ee
The statement of the Tannery's theorem, which is nothing but a special case of Lebesgue's dominated convergence theorem, is the following. \emph{For all $n\in\N$, let $\{a_m(n)\}_{m\in\N}\subset\R$ be a sequence. Suppose that the limit
\bb
\lim_{n\rightarrow \infty }a_{m}(n)\,
\ee
exists for all $m\in\N$. If there exists a sequence $\{M_m\}_{m\in\N}\subset\R$ such that $\sum_{m=0}^{\infty}M_m<\infty$ and $$|a_m(n)|\le M_m$$ for all $m,n\in\N$, then}
\bb
\lim\limits_{n\rightarrow\infty}\sum _{m=0}^{\infty }a_{m}(n)=\sum_{m=0}^{\infty}\lim_{n\rightarrow \infty }a_{m}(n)\,.
\ee
By setting
\begin{equation}
	\chi_{[0,n-k]}(m)\coloneqq\begin{cases}
	1 & \text{if $m\in[0,n-k]$,} \\
	0 & \text{otherwise.}
	\end{cases}
\end{equation}
one obtains 
\bb
&\lim\limits_{n\rightarrow\infty}\sum_{m=0}^{n-k}\left(1-\frac{c}{n}\right)^{-m}\left[c\sqrt{N(N+1)}\right]^{2m}\frac{\binom{n}{m+k}\binom{n-k}{m}}{n^{2m+k}} \\&=\lim\limits_{n\rightarrow\infty}\sum_{m=0}^{\infty}\chi_{[0,n-k]}(m)\left(1-\frac{c}{n}\right)^{-m} \left[c\sqrt{N(N+1)}\right]^{2m}\frac{\binom{n}{m+k}\binom{n-k}{m}}{n^{2m+k}}\,.
\ee
Let us check whether the hypotheses of Tannery's theorem are fulfilled.
Thanks to the inequality $\binom{a}{b}\le\frac{a^b}{b!}$, valid for all integers $a\ge b\ge0$, and to the fact that $(1-\frac{c}{n})^{-m}\le 2^m$, valid for sufficiently large $n$ (more precisely, as soon as $n\ge2c$), the general term of the series (which is non-negative) satisfies
\bb
&\chi_{[0,n-k]}(m)\left(1-\frac{c}{n}\right)^{-m}\left[c\sqrt{N(N+1)}\right]^{2m}\frac{\binom{n}{m+k}\binom{n-k}{m}}{n^{2m+k}} \le \frac{\left[c\sqrt{2N(N+1)}\right]^{2m}}{(m+k)!m!}\,,
\ee
for sufficiently large $n$. Since~\footnote{Notice that $\sum_{m=0}^{\infty}\frac{\left[c\sqrt{2N(N+1)}\right]^{2m}}{(m+k)!m!}<\exp\left[2N(N+1)c^2\right]<\infty$.} 
\bb
\sum_{m=0}^{\infty}\frac{\left[c\sqrt{2N(N+1)}\right]^{2m}}{(m+k)!m!}<\infty\,,
\ee
Tannery's theorem guarantees that
\bb
&\lim\limits_{n\rightarrow\infty}\sum_{m=0}^{n-k}\left(1-\frac{c}{n}\right)^{-m}\left[c\sqrt{N(N+1)}\right]^{2m}\frac{\binom{n}{m+k}\binom{n-k}{m}}{n^{2m+k}} \\&=\sum_{m=0}^{\infty}\lim\limits_{n\rightarrow\infty}\chi_{[0,n-k]}(m)\left(1-\frac{c}{n}\right)^{-m} \left[c\sqrt{N(N+1)}\right]^{2m}\frac{\binom{n}{m+k}\binom{n-k}{m}}{n^{2m+k}} \\&=\sum_{m=0}^{\infty}\lim\limits_{n\rightarrow\infty}\left(1-\frac{c}{n}\right)^{-m}\left[c\sqrt{N(N+1)}\right]^{2m}\frac{\binom{n}{m+k}\binom{n-k}{m}}{n^{2m+k}}\\&=\sum_{m=0}^\infty\frac{\left[c\sqrt{N(N+1)}\right]^{2m}}{(m+k)!m!}\,.\label{speroult}
\ee
As a consequence, \eqref{seriess} implies that
\begin{equation}\label{primaeq}
q_{-k}(N,c)= e^{-c\left[2N+1\right]}\left[c(N+1)\right]^{k}\sum_{m=0}^{\infty}\frac{\left[c\sqrt{N(N+1)}\right]^{2m}}{(m+k)!m!}\,.
\end{equation}
By using \eqref{besself}, we arrive at
\begin{equation}
q_{-k}(N,c)= e^{-c\left[2N+1\right]}\left(\frac{N+1}{N}\right)^{k/2}I_k\left(2c\sqrt{N(N+1)}\right)\,.
\end{equation}
Analogously, for $k\ge0$ it holds that:
\bb
&q_{k}(N,c) \\
&=\lim\limits_{n\rightarrow\infty}\frac{\sum_{m=0}^n\left(\frac{c}{n}\right)^{2m+k}\left(1-\frac{c}{n}\right)^{n-m}N^{m+k}(N+1)^m\binom{n}{m}\binom{n+k}{n-m}}{(1+N\frac{c}{n})^{2n+k+1}} \\
&=\lim\limits_{n\rightarrow\infty}\frac{1}{(1+N\frac{c}{n})^{2n+k+1}}\left(1-\frac{c}{n}\right)^{n}\left(cN\right)^k \sum_{m=0}^{n}\left(1-\frac{c}{n}\right)^{-m}\left[c\sqrt{N(N+1)}\right]^{2m}\frac{\binom{n}{m}\binom{n}{m+k}}{n^{2m+k}} \\
&=e^{-c\left[2N+1\right]}\left(cN\right)^{k}\sum_{m=0}^{\infty}\frac{\left[c\sqrt{N(N+1)}\right]^{2m}}{(m+k)!m!} \\
&=e^{-c\left[2N+1\right]}\left(\frac{N}{N+1}\right)^{k/2}I_k\left(2c\sqrt{N(N+1)}\right)\,.&
\ee
In summary, we have shown that for all $k\in\mathbb{Z}$ it holds that
\begin{equation}
q_k(N,c)=e^{-c\left[2N+1\right]}\left(\frac{N}{N+1}\right)^{k/2}I_{|k|}\left(2c\sqrt{N(N+1)}\right)\,.
\end{equation}
Now, let us define the following state on $\HH_{\text{ext}}$:
\begin{equation}
\rho_{q(N,c)}\coloneqq\sum_{k=-\infty}^{\infty}q_{k}(N,c)\ketbra{k}\,.
\end{equation}
\eqref{shiftedState} and \eqref{qk} guarantee that the sequence of density operators $\left\{T_{-n}\Phi_{\frac{c}{n},\ketbrasub{n}}(\tau_N)T_{-n}^\dagger\right\}_{n\in\N}$ weakly converges to $\rho_{q(N,c)}$. Consequently, we can apply~\cite[Lemma 4.3]{Davies1969}, which states that: \emph{if a sequence of density operators converges to another density operators in the weak operator topology, then the convergence is in trace norm}.  Hence,
$\left\{T_{-n}\Phi_{\frac{c}{n},\ketbrasub{n}}(\tau_N)T_{-n}^\dagger\right\}_{n\in\N}$ converges to $\rho_{q(N,c)}$ in trace norm i.e.
\begin{equation}\label{trconvq}
\lim\limits_{n\rightarrow\infty}\left\| T_{-n}\Phi_{\frac{c}{n},\ketbrasub{n}}(\tau_N)T_{-n}^\dagger-\rho_{q(N,c)}\right\|_1=0\,.
\end{equation}
Analogously, we define the probability distribution $\{p_k(N,c)\}_{k\in\mathbb{Z}}$ over the alphabet $\mathbb{Z}$ as
\begin{equation}\label{pk}
p_k(N,c)\coloneqq
\begin{cases}
\lim\limits_{n\rightarrow\infty}P_k\left(N,n,1-\frac{c}{n}\right), & \text{if $k\ge 0$,} \\
0, & \text{otherwise.}
\end{cases}
\end{equation}
Let $k\ge0$.  Analogously to what has been done previously i.e.~starting from \eqref{simpler}, carrying out the calculations and expanding for large $n$, one obtains that:
\bb
&p_k(N,c) \\
&\quad =\lim\limits_{n\rightarrow\infty}\frac{1}{\left(N+1-N\frac{c}{n}\right)^{k+n+1}}\sum_{m=n-k}^n\left(1-\frac{c}{n}\right)^{2m+k-n}  \left(\frac{c}{n}\right)^{n-m}N^{k+m-n}(N+1)^m\binom{n}{m}\binom{k}{n-m} \\&=\lim\limits_{n\rightarrow\infty}\frac{1}{(N+1)^{k+n+1}\left(1-\frac{Nc}{(N+1)n}\right)^{k+n+1}}\left(1-\frac{c}{n}\right)^{n+k}  \sum_{m=0}^k\left(1-\frac{c}{n}\right)^{-m}		\left(\frac{c}{n}\right)^{m}N^{k-m}(N+1)^{n-m}\binom{n}{m}\binom{k}{m}	 \\
&=\frac{N^k}{(N+1)^{k+1}}e^{-c/(N+1)}\sum_{m=0}^k	\left(\frac{c}{N(N+1)}\right)^{m}\frac{1}{m!}\binom{k}{m} \\
&=\frac{N^k}{(N+1)^{k+1}}e^{-c/(N+1)}L_k\left(-\frac{c}{N(N+1)}\right)\,.
\ee
Moreover, let us define
\begin{equation}
	\rho_{p(N,c)}\coloneqq \sum_{k=0}^{\infty}p_{k}(N,c)\ketbra{k}\,.
\end{equation}
Since $\left\{\Phi_{1-\frac{c}{n},\ketbrasub{n}}(\tau_N)\right\}_{n\in\N}$ weakly converges to $\rho_{p(N,c)}$, then it converges to $\rho_{p(N,c)}$ in trace norm~\cite[Lemma 4.3]{Davies1969} i.e.
\begin{equation}\label{trconvp}
\lim\limits_{n\rightarrow\infty}\left\| \Phi_{1-\frac{c}{n},\ketbrasub{n}}(\tau_N)-\rho_{p(N,c)}\right\|_1=0\,.
\end{equation}
From \eqref{conscons2} we have that:
\bb \label{conscons3}
&\liminf_{n\rightarrow\infty} I_{\text{coh}}\left(\Phi_{\frac{c}{n},\ketbrasub{n}},\tau_N\right) = \liminf_{n\rightarrow\infty}\left[S\left(T_{-n}\Phi_{\frac{c}{n},\ketbrasub{n}}(\tau_N)T_{-n}^\dagger\right)  -S\left(\Phi_{1-\frac{c}{n},\ketbrasub{n}}(\tau_N)\right)\right]\,.
\ee
At this point we would like to show that we can lower bound the expression~\eqref{conscons3} by taking the limit inside the function $S$. For this to be a legal move, we would need the entropy to be a continuous function of the states we consider. However, although in the finite-dimensional scenario Fannes' inequality~\cite[Chapter~11]{NC} guarantees the continuity in trace norm of the von Neumann entropy~\cite{Fannes1973}, in the infinite-dimensional setting the continuity does not hold any more. However, the von Neumman entropy is still lower semi-continuous in trace norm~\cite[Theorem 11.6]{HOLEVO-CHANNELS-2}.
Continuity of the von Neumann entropy is restored when restricting to quantum systems satisfying the Gibbs hypothesis (i.e.\ quantum systems having finite partition function) and to states with 
bounded energy, as established by~\cite[Lemma 11.8]{HOLEVO-CHANNELS-2}.
First, we are going to apply~\cite[Lemma 11.8]{HOLEVO-CHANNELS-2} to the sequence $\left\{\Phi_{1-\frac{c}{n},\ketbrasub{n}}(\tau_N)\right\}_{n\in\N}$. Let us check whether its hypotheses are fulfilled.
\begin{itemize}
	\item The partition function of the quantum harmonic oscillator is finite for all $\beta>0$, indeed:
	$$\Tr_S\left[e^{-\beta a^\dagger a}\right]=\frac{1}{1-e^{-\beta}}<\infty\quad\text{for all }\beta>0\,.$$
	\item The sequence $\left\{\Phi_{1-\frac{c}{n},\ketbrasub{n}}(\tau_N)\right\}_{n\in\N}\subseteq\mathfrak{S}(\HH_S)$ has 
	bounded energy. Indeed, Lemma~\ref{LemmaEner} implies that for all $n\in\N$ it holds that 
	$$\Tr_S\left[\Phi_{1-\frac{c}{n},\ketbrasub{n}}(\tau_N)\,a^\dagger a\right]=\left(1-\frac{c}{n}\right)N+c\le N+c\,.$$
\end{itemize}
As a consequence, we can apply~\cite[Lemma 11.8]{HOLEVO-CHANNELS-2}. By using also \eqref{trconvp}, we deduce that
\bb
\lim\limits_{n\rightarrow\infty}S\left(\Phi_{1-\frac{c}{n},\ketbrasub{n}}(\tau_N)\right)=S\left(\rho_{p(N,c)}\right)=H\left(p(N,c)\right)\,.
\ee
Then, \eqref{conscons3} implies 
\bb\label{conscons4}
&\liminf_{n\rightarrow\infty} I_{\text{coh}}\left(\Phi_{\frac{c}{n},\ketbrasub{n}},\tau_N\right) =\liminf_{n\rightarrow\infty}S\left(T_{-n}\Phi_{\frac{c}{n},\ketbrasub{n}}(\tau_N)T_{-n}^\dagger\right)  -H\left(p(N,c)\right)\,.
\ee
Second, for the sequence $\left\{T_{-n}\Phi_{\frac{c}{n},\ketbrasub{n}}(\tau_N)T_{-n}^\dagger\right\}_{n\in\N}\subseteq\mathfrak{S}(\HH_{\text{ext}})$, we apply the lower semi-continuity of the von Neumann entropy and \eqref{trconvq} to deduce that
\bb
    \liminf\limits_{n\rightarrow\infty}S\left(T_{-n}\Phi_{\frac{c}{n},\ketbrasub{n}}(\tau_N)T_{-n}^\dagger\right)&\ge S(\rho_{q(N,c)}) =H\left(q(N,c)\right) \,.
\ee
Hence, by substituting in \eqref{conscons4}, we finally obtain \eqref{entdifflow}.
\end{proof}

\section{Control of the environment}
\begin{Def}
Let $a$ the annihilation operator of a single-mode system $\HH_S$, let $N>0$, and let $\Delta:\mathfrak{S}(\HH_S)\mapsto\mathfrak{S}(\HH_S)$. One defines the energy-constrained diamond norm of $\Delta$ with energy constraint $N$ as~\cite{PLOB, Shirokov2018, VV-diamond}
\begin{equation}\label{ednorm}
	\left\|\Delta\right\|_{\diamond N}\coloneqq\sup_{\rho\in\mathfrak{S}(\HH_S\otimes\HH_C)\text{ : }\Tr\left[\rho\, a^\dagger a\right]\le N} \left\|(\Delta\otimes I_C)\rho\right\|_1\text{ ,}
\end{equation}
where the $\sup$ is taken also over the ancilla systems $\HH_C$.
\end{Def}

\begin{lemma}{\cite[Theorem 9]{VV-diamond}}\label{continuityBound}
Let $\Phi_1,\Phi_2:\mathfrak{S}(\HH_S)\mapsto\mathfrak{S}(\HH_S)$ be quantum channels. Suppose there exist $\alpha$, $N_0$ $\in\mathbb{R}$ such that for all $\rho\in\mathfrak{S}(\HH_S)$ it holds that
\begin{equation}
    \Tr\left[a^\dagger a\, \Phi_i(\rho)\right]\le \alpha \Tr [a^\dagger a\,\rho] + N_0\qquad \forall\ i=1,2\,.
\end{equation}
Quantum channels satisfying such a constraint are said to be `energy limited'.

Suppose that there exists $\varepsilon\in(0,1)$ such that for all $N>0$ it holds that
\begin{equation}
    \frac{1}{2}\left\| \Phi_1-\Phi_2\right\|_{\diamond N}\le \varepsilon\,.
\end{equation}
Then for all $N>0$ it holds that
\bb \label{countboundq}
&\left|Q\left(\Phi_1,N\right)-Q\left(\Phi_2,N\right)\right| \le 56\sqrt{\varepsilon}\, g\left(4\frac{\alpha N + N_0}{\sqrt{\varepsilon}}\right)+6g\left(4\sqrt{\varepsilon}\right)\,,
\ee
where $g(x)\coloneqq (x+1)\log_2(x+1) - x\log_2 x$. Similar statements can be proved for the energy-constrained entanglement-assisted capacities.

In particular, $Q(\cdot\,,N)$ and $C_{\text{ea}}\left(\cdot\,,N\right)$ are continuous with respect to the EC diamond norm over the subset of energy-limited quantum channels.
\end{lemma}

\begin{lemma}{\cite[Lemma 14]{Die-Hard-2-PRA}}\label{Stability}
	The energy-constrained quantum capacity and the energy-constrained entanglement-assisted classical capacity of the general attenuators are continuous with respect to the environment state, over the subset of environment states having a finite mean photon number.

	More explicitly, let $\sigma\in\mathfrak{S}(\HH_E)$ be an environment state such that $\langle b^\dagger b\rangle_\sigma\coloneqq \Tr[b^\dagger b\,\sigma]<\infty$ and let $\{\sigma_k\}_{k\in\mathbb{N}}\subset\mathfrak{S}(\HH_E)$ be a sequence of environment states such that $$\lim\limits_{k\rightarrow\infty}\|\sigma_k-\sigma\|_1=0\,,$$
	then for all $N>0$ and $\lambda\in[0,1]$ it holds that:
	\bb
	\lim\limits_{k\rightarrow\infty}Q\left(\Phi_{\lambda,\sigma_k},N\right)&=Q\left(\Phi_{\lambda,\sigma},N\right)\,,\\
	\lim\limits_{k\rightarrow\infty}C_{\text{ea}}\left(\Phi_{\lambda,\sigma_k},N\right)&=C_{\text{ea}}\left(\Phi_{\lambda,\sigma},N\right)\,.
	\ee
\end{lemma}
\begin{lemma}\label{lemma_dist}
Let $\nu\ge0$, $n\in\N$, $\lambda\in(0,1)$, $k\in\N$, and $\rho_{\lambda,n}^{(k)}\coloneqq\ketbra{n,\lambda}_{S_1S_2\ldots S_k}$, where:
\bb
\ket{n,\lambda}_{S_1S_2\ldots S_k}&\coloneqq U^{(S_{1}S_2)}_{\frac{1-\lambda}{1-\lambda^2}}\ldots U^{(S_{k-1}S_k)}_{\frac{1-\lambda}{1-\lambda^k}}\ket{0}_{S_1}\ldots \ket{0}_{S_{k-1}}\ket{n}_{S_k}. \label{trigger_signals}
\ee
By defining
\bb\label{achievable_appk}
\sigma_{\lambda,\nu,n,k}\coloneqq\Tr_{S_1S_2\ldots S_k} & \left[U_{\lambda}^{(S_kE)}\ldots U_{\lambda}^{(S_1E)}\rho_{\lambda,n}^{(k)}\otimes\tau_\nu \left(U_{\lambda}^{(S_1E)}\right)^\dagger\ldots     \left(U_{\lambda}^{(S_kE)}\right)^\dagger\right]\,,
\ee
it holds that
\begin{align}
&\frac{1}{4}\left\|\sigma_{\lambda,\nu,n,k}-\ketbra{n} \right\|_1^2\le  \left[n^2-(4n+1)\langle b^\dagger b\rangle_{\tau_\nu}+\langle (b^\dagger b)^2\rangle_{\tau_\nu}-n\right]\lambda^{2k}+\left[(2n+1)\langle b^\dagger b\rangle_{\tau_\nu}+n\right]\lambda^{k}\,,\label{distance_from_Fock2}
\end{align}
where $\langle b^\dagger b\rangle_{\tau_\nu}=\nu$ and $\langle (b^\dagger b)^2\rangle_{\tau_\nu}=\nu(2\nu+1)$.
\end{lemma}
\begin{proof}
The lemma follows from \cite[Eq.~189]{Die-Hard-2-PRA}.
\end{proof}
\begin{thm*}[4]
Let $\nu\ge0$, $n\in\N$, and $\lambda\in(0,1)$. There exists a sequence of density operators $\{\rho_{\lambda,n}^{(k)}\}_{k\in\N}$, with $\rho_{\lambda,n}^{(k)}\in\mathfrak{S}(\HH_{S_{1}}\otimes\ldots\otimes\HH_{S_{k}})$, such that 
\bb\label{limit_env}
\lim\limits_{k\to\infty}\|\sigma_{\lambda,\nu,n,k}-\ketbra{n}\|_1=0\,,
\ee
where
\bb\label{achievable_app}
\sigma_{\lambda,\nu,n,k}\coloneqq\Tr_{S_1S_2\ldots S_k} & \left[U_{\lambda}^{(S_kE)}\ldots U_{\lambda}^{(S_1E)}\rho_{\lambda,n}^{(k)}\otimes\tau_\nu \left(U_{\lambda}^{(S_1E)}\right)^\dagger\ldots     \left(U_{\lambda}^{(S_kE)}\right)^\dagger\right]\,.
\ee
Furthermore, let $\lambda\in (0,1/2)$ and set $n=n_\lambda\in\mathbb{N}$ with $1/\lambda-1\le n_\lambda\le 1/\lambda$, then for $k$ sufficiently large it holds that 
\bb
Q(\Phi_{\lambda,\sigma_{\lambda,\nu,n_\lambda,k}})>0\,.
\ee
In addition, if Conjecture~2 holds, then for all $N>0$, $\lambda\in(0,1/2)$, and for $\bar{n}\in\N$ sufficiently large, it holds that
\bb\label{limit_k_Cea}
    \lim\limits_{k\to\infty}C_{\text{ea}}\left(\Phi_{\sigma_{\lambda,\nu,\bar{n},k}},N\right)&>C\left(\Id,N\right)\,,\\
    \lim\limits_{k\to\infty}Q_{\text{ea}}\left(\Phi_{\lambda,\sigma_{\lambda,\nu,\bar{n},k}},N\right)&>Q\left(\Id,N\right)/2\,,\\
    \lim\limits_{k\to\infty}Q\left(\Phi_{\lambda,\sigma_{\lambda,\nu,\bar{n},k}},N\right)&>0\,.
\ee
\end{thm*}
\begin{proof}
\eqref{limit_env} is a consequence of Lemma~\ref{lemma_dist}.

Let $\lambda\in(0,1/2)$ and $n_\lambda\in \N$ with $1/\lambda-1\le n_\lambda\le 1/\lambda$. Lemma~\ref{Stability}, \eqref{limit_env},  and Theorem~\ref{diehard_th} guarantee that
\begin{equation}
    \lim\limits_{k\rightarrow\infty}Q\left(\Phi_{\lambda,\sigma_{\lambda,\nu,n_\lambda,k}},1/2\right)=Q\left(\Phi_{\lambda,\ketbrasub{n_\lambda} },1/2\right)>0\,.
\end{equation}
Since $Q\left(\Phi\right)\ge Q\left(\Phi,1/2\right)$ for any quantum channel $\Phi$, we deduce that $Q\left(\Phi_{\lambda,\sigma_{\lambda,\nu,n_\lambda,k}}\right)>0$ for $k$ sufficiently large.

In addition, under the assumption that Conjecture~2 is valid, for all $\lambda\in(0,1/2)$, $N>0$, and for $n\in\N$ sufficiently large it holds that 
\begin{align}
C_{\text{ea}}\left(\Phi_{\lambda,\ketbrasub{n}},N\right) &> C\left(\Id,N\right)\,, \label{SpecialCea} \\
Q_{\text{ea}}\left(\Phi_{\lambda,\ketbrasub{n}},N\right) &> Q\left(\Id,N\right)/2\,, \label{SpecialQea} \\
Q\left(\Phi_{\lambda,\ketbrasub{n}},N\right) &> 0\,. \label{SpecialQ}
\end{align}
This fact, together with \eqref{limit_env} and Lemma~\ref{Stability}, would imply the validity of the limits in \eqref{limit_k_Cea}.
\end{proof}

\begin{lemma}\label{lemma_bound_ene}
For all $\lambda\in [0,1]$, $\rho\in\mathfrak{S}(\HH_S)$, $\sigma\in\mathfrak{S}(\HH_E)$, it holds that
\begin{align}
&\Tr_S[\Phi_{\lambda,\sigma}(\rho)\,a^\dagger a]\le \left(\lambda+\frac{1}{2}\sqrt{\lambda(1-\lambda)\langle b^\dagger b\rangle_\sigma}\right)\langle a^\dagger a \rangle_\rho+\left[(1-\lambda)\langle b^\dagger b\rangle_\sigma+2\sqrt{\lambda(1-\lambda)\langle b^\dagger b\rangle_\sigma}\right]\,.
\end{align}
\end{lemma}
\begin{proof}
By using Lemma~\eqref{LemmaEner}, we have that
\bb \label{enerbound}
&\Tr_S\left[\Phi_{\lambda,\sigma}(\rho)\,a^\dagger a\right] \le\lambda \langle a^\dagger a \rangle_\rho +(1-\lambda)\langle b^\dagger b\rangle_\sigma+2\sqrt{\lambda(1-\lambda)}|\langle a\rangle_\rho| |\langle b^\dagger\rangle_\sigma|\,.
\ee
The Cauchy--Schwarz inequality for the Hilbert--Schmidt inner product yields 
\bb
&|\langle a\rangle_\rho|=|\langle a^\dagger\rangle_\rho|=|\Tr[\sqrt{\rho}\sqrt{\rho} a^\dagger]|\le\sqrt{\Tr\rho}\sqrt{\Tr[a\rho a^\dagger]} =\sqrt{\langle a^\dagger a \rangle_\rho}\,, \\
&|\langle b^\dagger\rangle_\sigma|\le  \sqrt{\langle b^\dagger b\rangle_\sigma}\,.
\ee
By inserting this into \eqref{enerbound} and by noting that 
\bb
\sqrt{\langle a^\dagger a \rangle_\rho}\le \frac{\langle a^\dagger a \rangle_\rho}{4}+1\,,
\ee
one obtains that
\begin{align}
\Tr_S[\Phi_{\lambda,\sigma}(\rho)\,a^\dagger a] &\le \lambda \langle a^\dagger a \rangle_\rho +(1-\lambda)\langle b^\dagger b\rangle_\sigma+2\sqrt{\lambda(1-\lambda)\langle b^\dagger b\rangle_\sigma}\sqrt{\langle a^\dagger a \rangle_\rho}\nonumber\\&\le \left(\lambda+\frac{1}{2}\sqrt{\lambda(1-\lambda)\langle b^\dagger b\rangle_\sigma}\right)\langle a^\dagger a \rangle_\rho\nonumber+\left[(1-\lambda)\langle b^\dagger b\rangle_\sigma+2\sqrt{\lambda(1-\lambda)\langle b^\dagger b\rangle_\sigma}\right]\,.
\end{align}
\end{proof}

\begin{lemma}\label{diamondtrace}
For all $\lambda\in[0,1]$, $N>0$, $\sigma,\sigma_0\in\mathfrak{S}(\HH_E)$, it holds that
\begin{equation}\label{eq_diamondtrace}
\left\|\Phi_{\lambda,\sigma}-\Phi_{\lambda,\sigma_0}\right\|_{\diamond N}\le\left\|\sigma-\sigma_0\right\|_1\,.
\end{equation}
\end{lemma}
\begin{proof}
Take an arbitrary Hilbert space $\HH_C$ and $\rho\in\mathfrak{S}(\HH_S\otimes\HH_C)$ such that $\Tr\left[\rho\,a^\dagger a\right]\le N$. By applying~\cite[Lemma 23]{QCLT}, we have that
\bb\label{diamond_dist_gen}
&\left\|\left(\left(\Phi_{\lambda,\sigma}-\Phi_{\lambda,\sigma_0}\right)\otimes I_C\right)(\rho)\right\|_1 \le \left\|\sigma-\sigma_0\right\|_1\,.
\ee
\eqref{diamond_dist_gen} can be derived from the contractivity of the trace norm under partial traces and from the invariance of the trace norm under unitary transformations. \ref{eq_diamondtrace} follows from \eqref{diamond_dist_gen}.
\end{proof}

\begin{thm*}[5]
Let $\nu\ge0$, $\lambda\in(0,1)$ and $\rho_{\lambda}\coloneqq U_{\frac{1}{1+\lambda}}^{(S_1S_2)}\ket{0}_{S_1}\ket{n_\lambda}_{S_2}$,
	with $n_\lambda\in\mathbb{N}$ such that $1/\lambda-1\le n_\lambda\le 1/\lambda$.
By defining
\bb\label{achievable_app2}
\sigma_{\lambda,\nu}\coloneqq\Tr_{S_1S_2} & \left[U_{\lambda}^{(S_2E)}U_{\lambda}^{(S_1E)}\rho_{\lambda}\otimes\tau_\nu \left(U_{\lambda}^{(S_1E)}\right)^\dagger     \left(U_{\lambda}^{(S_2E)}\right)^\dagger\right]\,,
\ee
for $\lambda>0$ sufficiently small it holds that 
\bb\label{die_hard_achiev}
Q\left(\Phi_{\lambda,\sigma_{\lambda,\nu}}\right)\ge Q\left(\Phi_{\lambda,\sigma_{\lambda,\nu}},1/2\right)\ge c\,,
\ee
where $c>0$ is a fixed positive constant. 
\end{thm*}

\begin{proof}
Lemma~\ref{lemma_dist} implies that
\bb
\frac{1}{4}\left\|\sigma_{\lambda,\nu}-\ketbra{n_\lambda}\right\|_1^2 &\le \lambda^4 n_\lambda^2+\lambda^4 \langle (b^\dagger b)^2\rangle_{\tau_\nu} +\lambda^2(2n_\lambda+1)\langle b^\dagger b\rangle_{\tau_\nu}+\lambda^2n_\lambda\,.
\ee
By using $n_\lambda\le 1/\lambda$ and $\lambda\le1$, one obtains
\begin{equation}\label{tracenorm2}
\frac{1}{2}\left\|\sigma_{\lambda,\nu}-\ketbra{n_\lambda}\right\|_1\le k_0 \sqrt{\lambda}\,,
\end{equation}
where
\begin{equation}
k_0\coloneqq \left[2+3\langle b^\dagger b\rangle_{\tau_\nu}+\langle(b^\dagger b)^2\rangle_{\tau_\nu}\right]^{1/2}\,.
\end{equation}
In order to prove \eqref{die_hard_achiev}, we want to exploit the continuity bound with respect to the energy-constrained diamond norm of Lemma~\ref{continuityBound}. 
First, we need to find $\alpha(\lambda)$ and $N_0(\lambda)$ such that for any state $\rho$ it holds that
\bb
\Tr\left[a^\dagger a\, \Phi_{\lambda,\sigma_{\lambda,\nu}}(\rho)\right]&\le \alpha(\lambda) \Tr\left[a^\dagger a\,\rho\right]+N_0(\lambda)\,,\\
\Tr\left[a^\dagger a\, \Phi_{\lambda,\ketbrasub{n_\lambda}}(\rho)\right]&\le \alpha(\lambda) \Tr\left[a^\dagger a\,\rho\right]+N_0(\lambda)\,.
\ee
By using Lemma~\ref{lemmatrasf}, one can show that
\bb
    \langle b^\dagger b\rangle_{\sigma_{\lambda,\nu}}&=\Tr_{S_1S_2E} \left[b^\dagger b\,U_{\lambda}^{(S_2E)}U_{\lambda}^{(S_1E)}\rho_{\lambda}\otimes\tau_\nu \left(U_{\lambda}^{(S_1E)}\right)^\dagger     \left(U_{\lambda}^{(S_2E)}\right)^\dagger\right] \\&=\Tr_{S_1S_2 E}\left[\left(-\sqrt{1-\lambda^2}h_{\lambda,2}+\lambda b\right)^\dagger \left(-\sqrt{1-\lambda^2}h_{\lambda,2}+\lambda b\right) \rho_\lambda\otimes \tau_\nu\right]\\&=(1-\lambda^2)\Tr_{S_1S_2}\left[\left({h_{\lambda,2}}\right)^\dagger{h_{\lambda,2}} \, \rho_\lambda\right]+\lambda^2\Tr_E[b^\dagger b \, \tau_\nu]\,,
\ee
where we have defined
\bb
h_{\lambda,2}\coloneqq\sqrt{\frac{1-\lambda}{1-\lambda^2}}\left(\sqrt{\lambda}a_1+a_2\right)=U^{(S_{1}S_2)}_{\frac{1}{1+\lambda}}a_2\left(U^{(S_{1}S_2)}_{\frac{1}{1+\lambda}}\right)^\dagger\,.
\ee
Hence, by using that $\rho_{\lambda}\coloneqq U_{\frac{1}{1+\lambda}}^{(S_1S_2)}\ket{0}_{S_1}\ket{n_\lambda}_{S_2}$, one obtains
\bb
    \langle b^\dagger b\rangle_{\sigma_{\lambda,\nu}}=(1-\lambda^2)n_\lambda+\lambda^2\nu\,.
\ee
Notice that 
\begin{equation}
\langle b^\dagger b\rangle_{\ketbrasub{n_\lambda}}=n_\lambda \le 1/\lambda\,,
\end{equation}
and that for $\lambda>0$ sufficiently small it holds that
\bb
\langle b^\dagger b \rangle_{\sigma_{\lambda,\nu}}=  n_\lambda-\lambda^2\left(n_\lambda-\nu\right)\le n_\lambda\le 1/\lambda\,.
\ee
As a consequence, Lemma~\ref{lemma_bound_ene} guarantees that for $\lambda>0$ sufficiently small we can take 
\bb
    \alpha(\lambda)&\equiv 1\,,\\
    N_0(\lambda)&\equiv \frac{1}{\lambda}+2\,.
\ee
Second, we need to find an upper bound on the energy-constrained diamond distance between $\Phi_{\lambda,\ketbrasub{n_\lambda}}$ and $\Phi_{\lambda,\sigma_{\lambda,\nu}}$. By using Lemma~\ref{diamondtrace} and \eqref{tracenorm2}, we obtain
\begin{equation}
\frac{1}{2}\left\|\Phi_{\lambda,\sigma_{\lambda,\nu}}-\Phi_{\lambda,\ketbrasub{n_\lambda}}\right\|_{\diamond N}\le k_0\sqrt{\lambda}\,,
\end{equation}
for all $N>0$.
By setting $\varepsilon(\lambda)\coloneqq k_0\sqrt{\lambda}$, we apply Theorem~\ref{continuityBound} to obtain
\bb
\left|Q\left(\Phi_{\lambda,\sigma_{\lambda,\nu}},N\right)-Q\left(\Phi_{\lambda,\ketbrasub{n_\lambda}},N\right)\right| &\le 56\sqrt{\varepsilon(\lambda)}\, g\left(4\frac{\alpha(\lambda) N + N_0(\lambda)}{\sqrt{\varepsilon(\lambda)}}\right)+6g\left(4\sqrt{\varepsilon(\lambda)}\right) \\&=56\sqrt{k_0}\lambda^{1/4}\text{ } g\left(4\frac{1+\lambda(N + 2)}{\sqrt{k_0}\lambda^{5/4}}\right)+6g\left(4\sqrt{k_0}\lambda^{1/4}\right)\,, \label{distquantum3}
\ee
for all $N>0$ and $\lambda>0$ sufficiently small.
By using that
\bb
\lim\limits_{x\rightarrow0^+}g(x)&=0\,,\\
\lim\limits_{x\rightarrow0+}x\log_2 x&=0\,,\\
g(x)&=\log_2(ex)+o(1)\quad\text{($x\rightarrow\infty$)}\,,
\ee
it follows that for all $N>0$ it holds that
\begin{equation}\label{lim_lambda_Q}
    \lim\limits_{\lambda\rightarrow 0^+}\left|Q\left(\Phi_{\lambda,\sigma_{\lambda,\nu}},N\right)-Q\left(\Phi_{\lambda,\ketbrasub{n_\lambda}},N\right)\right|=0\,.
\end{equation}
In addition, Theorem~\ref{diehard_th} establishes that there exists $\bar{\epsilon}>0$ and $\bar{c}>0$ such that for all $\lambda\in(0,1/2-\bar{\varepsilon})$ it holds that $Q\left(\Phi_{\lambda,\ketbrasub{n_\lambda}},1/2\right)\ge \bar{c}$. As a consequence, \eqref{lim_lambda_Q} guarantees that for $\lambda>0$ sufficiently small it holds that
\begin{equation}
    Q\left(\Phi_{\lambda,\sigma_{\lambda,\nu}},1/2\right)\ge Q\left(\Phi_{\lambda,\ketbrasub{n_\lambda}},1/2\right)-\frac{\bar{c}}{2}\ge \frac{\bar{c}}{2}\,.
\end{equation}
We conclude that the inequality
\begin{equation}
Q\left(\Phi_{\lambda,\sigma_{\lambda,\nu}}\right)\ge Q\left(\Phi_{\lambda,\sigma_{\lambda,\nu}},1/2\right)\ge \bar{c}/2\,,
\end{equation}
holds for $\lambda>0$ sufficiently small.
\end{proof}

\end{document}